\documentclass[10pt,journal,epsfig]{IEEEtran}

\usepackage[dvips]{graphicx}
\usepackage{graphicx}
\usepackage{amssymb}
\usepackage{cite}
\usepackage{subfigure}

\usepackage{amsmath}
\usepackage{multirow}

\begin{document}

\title{Channel Estimation for Millimeter Wave Multiuser MIMO Systems via PARAFAC Decomposition}

\author{Zhou Zhou, Jun Fang, Linxiao Yang, Hongbin Li, Zhi Chen,
and Shaoqian Li
\thanks{Zhou Zhou, Jun Fang, Linxiao Yang, Zhi Chen and Shaoqian Li are with the National Key Laboratory
of Science and Technology on Communications, University of
Electronic Science and Technology of China, Chengdu 611731, China,
Email: JunFang@uestc.edu.cn}
\thanks{Hongbin Li is
with the Department of Electrical and Computer Engineering,
Stevens Institute of Technology, Hoboken, NJ 07030, USA, E-mail:
Hongbin.Li@stevens.edu}
\thanks{This work was supported in part by the National Science
Foundation of China under Grants 61172114 and 61428103, and the
National Science Foundation under Grant ECCS-1408182. }}

\maketitle

\begin{abstract}
We consider the problem of uplink channel estimation for
millimeter wave (mmWave) systems, where the base station (BS) and
mobile stations (MSs) are equipped with large antenna arrays to
provide sufficient beamforming gain for outdoor wireless
communications. Hybrid analog and digital beamforming structures
are employed by both the BS and the MS due to hardware
constraints. We propose a layered pilot transmission scheme and a
CANDECOMP/PARAFAC (CP) decomposition-based method for joint
estimation of the channels from multiple users (i.e. MSs) to the
BS. The proposed method exploits the sparse scattering nature of
the mmWave channel and the intrinsic multi-dimensional structure
of the multiway data collected from multiple modes. The uniqueness
of the CP decomposition is studied and sufficient conditions for
essential uniqueness are obtained. The conditions shed light on
the design of the beamforming matrix, the combining matrix and the
pilot sequences, and meanwhile provide general guidelines for
choosing system parameters. Our analysis reveals that our proposed
method can achieve a substantial training overhead reduction by
employing the layered pilot transmission scheme. Simulation
results show that the proposed method presents a clear advantage
over a compressed sensing-based method in terms of both estimation
accuracy and computational complexity.
\end{abstract}

%effectiveness and superiority over other existing methods.


%Due to the strong pathloss in mm-Wave, hybrid analog/digital MIMO
%precoding and combining\cite{AlkhateebJianhua2014} has been
%proposed in transmission, which can provide sufficient
%multiplexing and array gain, and at the same time reduce the
%hardware cost and energy consumption.

%Thus an accurate channel estimation is crucial to the mmWave
%communication systems.

\section{Introduction}
Millimeter-wave (mmWave) communication is a promising technology
for future 5G cellular networks \cite{RanganRappaport14}. It has
the potential to offer gigabit-per-second data rates by exploiting
the large bandwidth available at mmWave frequencies. However,
communication at such high frequencies also suffers from high
attenuation and signal absorption \cite{SwindlehurstAyanoglu14}.
To compensate for the significant path loss, very large antenna
arrays can be used at the base station (BS) and the mobile station
(BS) to exploit beam steering to increase the link gain
\cite{AlkhateebMo14}. Due to the small wavelength at the mmWave
frequencies, the antenna size is very small and a large number of
array elements can be packed into a small area. Directional
precoding/beamforming with large antenna arrays is essential for
providing sufficient beamforming gain for mmWave communications.
On the other hand, the design of the precoding matrix requires
complete channel state information. Reliable mmWave channel
estimation, however, is challenging due to the large number of
antennas and the low signal-to-noise ratio (SNR) before
beamforming. The problem becomes exacerbated when considering
multi-user MIMO systems. Multi-user MIMO operation was advocated
in \cite{Marzetta10} which considers a single-cell time-division
duplex (TDD) scenario. The time-slot over which the channel can be
assumed constant is divided between uplink pilot transmission and
downlink data transmission. The BS, through channel reciprocity,
obtains an estimate of the downlink channel, and then generates a
linear precoder for transmitting data to multiple terminals
simultaneously. The time required for pilots, in this case,
increases linearly with the number of terminals served.

%In particular, for the uplink, the number of channel parameters
%grows linearly with the total number of antennas at the mobile
%stations (MSs) side. When the number of users increases, the

%To reduce the length of pilot symbols, the compressive channel
%sensing \cite{BajwaHaupt2010} has been proposed.

%is an attractive characteristic that

The sparse scattering nature of the mm-Wave channel can be
utilized to reduce the training overhead for channel estimation
\cite{AlkhateebyLeus15,AlkhateebAyach14,SchniterSayeed14}.
Specifically, it was shown \cite{AlkhateebyLeus15} that compressed
sensing-based methods achieve a significant training overhead
reduction via leveraging the poor scattering nature of mmWave
channels. In \cite{AlkhateebAyach14}, a novel hierarchical
multi-resolution beamforming codebook and an adaptive compressed
sensing method were proposed for channel estimation. The main idea
of adaptive compressed sensing-based channel estimation method is
to divide the training process into a number of stages, with the
training precoding used at each stage determined by the output of
earlier stages. Compared to the standard compressed sensing
method, the adaptive method is more efficient and yields better
performance at low signal-to-noise ratio (SNR). Nevertheless, this
performance improvement requires a feedback channel from the MS to
the BS, which may not be available before the communication
between the BS and the MS is established. Channel estimation and
precoding design for mmWave communications were also considered in
\cite{SchniterSayeed14}, where aperture shaping was used to ensure
a sparse virtual-domain MIMO channel representation.

%space-division multiple access in uplink or broadcasting in
%downlink for TDD system. Since the BS needs to estimate the CSI of
%the scheduled MSs to process the uplink detection or downlink
%precoding for simultaneously receiving or transmitting data from
%MSs during a coherence interval. When we utilize orthogonal
%pilots, the training overhead is proportional to the total number
%of antennas at the MSs side. As the number of MSs increases, the
%training overhead becomes unaffordable. Thus the number of MSs
%that can be served is limited by the coherence time. Therefore, if
%the CSI of more MSs could be acquired at BS for a given coherence
%interval, we can continue to improve the overall system throughput
%by spatial multiplexing.

%But in all these works, the proposed schemes is only considered
%for single-user or broadcast mm-Wave channel. Those methods can
%not more efficiently reduce the the uplink overhead in mm-Wave,
%since the overhead in multiple access has not been reduced.

%Note that although massive MIMO usually refer to conventional
%microwave frequency MIMO systems in which terminals are usually
%equipped with a smaller number of antennas, the idea of
%simultaneously serving a number of independent users still works
%for mmWave communication systems due to the poor scattering nature
%and the large antenna arrays at the BS.

In this paper, we consider the problem of multi-user uplink mmWave
channel estimation. Such a problem arises in multi-user massive
MIMO systems \cite{BradyBehdad13,AlkhateebLeus15b} where the BS,
via spatial multiplexing, simultaneously serves a number of
independent users sharing the same time-frequency bandwidth, and
thus requires to acquire the channel state information of multiple
users via uplink pilots (channel reciprocity is assumed). To
jointly estimate channels from multiple users to the BS, we
propose a layered pilot transmission scheme in which the training
phase consists of a number of frames and each frame is divided
into a number of sub-frames. In each sub-frame, users employ a
common beamforming vector to simultaneously transmit their
respective pilot symbols. With this layered transmission scheme,
the received signal at the BS can be represented as a third-order
tensor. We show that the third-order tensor admits a
CANDECOMP/PARAFAC (CP) decomposition and the channels can be
estimated from the CP factor matrices. Uniqueness of the CP
decomposition is studied. Our analysis shows that our proposed
method can achieve an additional training overhead reduction as
compared with a conventional scheme which separately estimates
multiple users' channels. We also compare our proposed method with
a compressed sensing-based method for joint channel estimation.
Simulation results show that the proposed method presents a clear
advantage over the compressed sensing-based method in terms of
both estimation accuracy and computational complexity.

%space-time-frequency coding for MIMO OFDM-CDMA
%\cite{FavierAlmeida14}

We note that mutlilinear tensor algebra, as a powerful tool, has
been widely used in a variety of applications in signal processing
and wireless communications, such as multiuser detection in
direct-sequence code-division multiple access (DS-CDMA)
\cite{SidiropoulosGiannakis00}, blind spatial signature estimation
\cite{RongVorobyov05}, two-way relaying MIMO communications
\cite{RoemerHaardt10}, etc. In particular, the uniqueness of CP
decomposition has proven useful in solving many array processing
problems from the multiple invariance sensor array processing
\cite{SidiropoulosBro00} to the detection and localization of
multiple targets in MIMO radar \cite{NionSidiropoulos10}. Another
important application is the multidimensional harmonic retrieval,
where significant improvements of parameter estimation accuracy
can be achieved by using multilinear algebra
\cite{HaardtRoemer08}. Recent years have seen a resurgence of
interest in tensor \cite{CichockiMandic15}, motivated by a number
of applications involving real-world multiway data.

%fueled by future wireless
%systems in the era of big data transmission \cite{BiZhang2015}.

The rest of the paper is organized as follows. In Section
\ref{sec:system-model}, we introduce the system model and a
layered pilot transmission scheme. Section
\ref{sec:tensor-overview} provides notations and basics on
tensors. In Section \ref{sec:proposed-method}, a tensor
decomposition-based method is developed for jointly estimating the
channels from multiple users to the BS. The uniqueness of the CP
decomposition is studied and sufficient conditions for the
uniqueness of the CP decomposition are derived in Section
\ref{sec:uniqueness-analysis}. A compressed sensing-based channel
estimation method is discussed in Section \ref{sec:cs-method}.
Computational complexity of the proposed method and the compressed
sensing-based method is analyzed in Section
\ref{sec:complexity-analysis}. Simulation results are provided in
Section \ref{sec:experiments}, followed by concluding remarks in
Section \ref{sec:conclusion}.

%\textit{Notations:} We use the following notations in this paper.
%$\cal A$ is a tensor, $\bf A$ is a matrix, $\bf a$ is a column
%vector and $a$ is a scalar. ${\bf A}^*$, ${\bf A}^H$, ${\bf A}^T$,
%${\bf A}^{-1}$, are the conjugate, conjugate transpose, transpose
%and inverse of $\bf A$ respectively. ${({\bf{a}})_I}$ are the
%entries of $\bf a$ with indices in the set $I$, ${{{a}}_{i,u}}$ is
%the $i$th element of ${\bf a}_u$, $A[:,k]$ and ${{\bf{A}}_{:,k}}$
%represent the $k$th column vector of matrix $\bf A$ and $diag(\bf a)$
%is a diagonal matrix whose diagonal entries are $\bf a$. $\left| {\bf{A}}
%\right|$ is a matrix whose each element takes the amplitude of the
%corresponding value in $\bf A$, $\max \bf A$ represents the
%maximum element of $\bf A$. $\bf I$ is identity matrix, and
%$\mathbb E[\cdot]$ is the expectation operation.

\section{System Model and Problem Formulation}\label{sec:system-model}
Consider a mmWave system consisting of a base station (BS) and $U$
mobile stations (MSs). We assume that hybrid analog and digital
beamforming structures (Fig. \ref{fig1}) are employed by both the
BS and the MS. The BS is equipped with $N_{\text{BS}}$ antennas
and $M_{\text{BS}}$ RF chains, and each MS is equipped with
$N_{\text{MS}}$ antennas and $M_{\text{MS}}$ RF chains. Since the
RF chain is expensive and power consuming, the number of RF chains
is usually less than the number of antennas, i.e.
$M_{\text{BS}}<N_{\text{BS}}$ and $M_{\text{MS}}<N_{\text{MS}}$.
We also assume $M_{\text{MS}}=1$, i.e. each user only transmits
one data stream.

In this paper, we consider the problem of estimating the uplink
mmWave channels from users to the BS. MmWave channels are expected
to have very limited scattering. Measurement campaigns in
dense-urban NLOS environments reveals that mmWave channels
typically exhibit only 3-4 scattering clusters, with relatively
little delay/angle spreading within each cluster
\cite{AkdenizLiu14}. Following \cite{AlkhateebAyach14}, we assume
a geometric channel model with $L_u$ scatterers between the $u$th
user and the BS. Under this model, the channel from the $u$th user
to the BS can be expressed as
\begin{align}
\boldsymbol{H}_u=\sum_{l=1}^{L_u}\alpha_{u,l}\boldsymbol{a}_{\text{BS}}(\theta_{u,l})\boldsymbol{a}_{\text{MS}}^T(\phi_{u,l})
\label{channel-model}
\end{align}
where $\alpha_{u,l}$ is the complex path gain associated with the
$l$th path of the $u$th user, $\theta_{u,l}\in [0,2\pi]$ and
$\phi_{u,l}\in [0,2\pi]$ are the associated azimuth angle of
arrival (AoA) and azimuth angle of departure (AoD), respectively,
$\boldsymbol{a}_{\text{BS}}(\theta_{u,l})$ and
$\boldsymbol{a}_{\text{MS}}(\phi_{u,l})$ denote the antenna array
response vectors associated with the BS and the MS, respectively.
In this paper, for simplicity, a uniform linear array is assumed,
though its extension to arbitrary antenna arrays is possible. The
steering vectors at the BS and the MS can thus be written as
follows respectively
\begin{align}
&\boldsymbol{a}_{\text{BS}}(\theta_{u,l}) \nonumber\\
\triangleq&\frac{1}{\sqrt{N_{\text{BS}}}}[1\phantom{0}e^{j(2\pi/\lambda)d\text{sin}(\theta_{u,l})}
\phantom{0}\ldots \phantom{0}
e^{j(N_{\text{BS}}-1)(2\pi/\lambda)d\text{sin}(\theta_{u,l})}]^T
\nonumber\\
&\boldsymbol{a}_{\text{MS}}(\phi_{u,l}) \nonumber\\
\triangleq&\frac{1}{\sqrt{N_{\text{MS}}}}[1\phantom{0}e^{j(2\pi/\lambda)d\text{sin}(\phi_{u,l})}
\phantom{0}\ldots \phantom{0}
e^{j(N_{\text{MS}}-1)(2\pi/\lambda)d\text{sin}(\phi_{u,l})}]^T
\nonumber
\end{align}
where $\lambda$ is the signal wavelength, and $d$ denotes the
distance between neighboring antenna elements.

%The steering vector for the MS can be written in a similar way,
%which is omitted here.

%The signal received at each MS can therefore be expressed as
%\begin{align}
%\boldsymbol{Y}=\boldsymbol{Q}^T\boldsymbol{H}\boldsymbol{P}\boldsymbol{S}+\boldsymbol{W}
%\label{received-signal}
%\end{align}
%where $\boldsymbol{Q}\triangleq
%[\boldsymbol{q}_1\phantom{0}\ldots\phantom{0}\boldsymbol{q}_Q]$ is
%the combining matrix, $\boldsymbol{P}\triangleq
%[\boldsymbol{p}_1\phantom{0}\ldots\phantom{0}\boldsymbol{p}_P]$ is
%the beamforming matrix, $\boldsymbol{H}$ denotes the mmWave
%channel matrix between the MS and the BS, and $\boldsymbol{S}$ is
%a diagonal matrix with the transmitted symbols $\{s_p\}_{p=1}^P$
%on its diagonal.

%by allowing users transmit orthogonal pilot symbols to the BS, for
%example, users can send their in an exclusively pre-specified time
%slot

%where a compressed sensing-based algorithm was developed to
%estimate the downlink channel by exploiting the sparse scattering
%nature of the mmWave channel

%in which users transmit their pilot symbols simultaneously

Note that the problem of single-user mmWave channel estimation has
been studied in \cite{AlkhateebAyach14,AlkhateebyLeus15}.
Specifically, to estimate the downlink channel, the BS employs $P$
different beamforming vectors at $P$ successive time frames, and
at each time frame, the MS uses $Q$ combining vectors to detect
the signal transmitted over each beamforming vector. By exploiting
the sparse scattering nature of mmWave channels, the problem of
estimating the mmWave channel can be formulated as a sparse signal
recovery problem and the training overhead can be considerably
reduced. The above method can also be used to solve our uplink
channel estimation problem if channels from users to the BS are
estimated separately. Nevertheless, we will show that a joint
estimation (of multiusers' channels) scheme may lead to an
additional training overhead reduction.

We first propose a layered pilot transmission scheme which is
elaborated as follows. The training phase consists of $T$
consecutive frames, and each frame is divided into $T'$
sub-frames. In each sub-frame $t'=1,\ldots,T'$, users employ a
common beamforming vector $\boldsymbol{p}_{t'}$ to simultaneously
transmit their respective pilot symbols $s_{u,t}$, where $s_{u,t}$
denotes the pilot symbol used by the $u$th user at the $t$th
frame. At the BS, the transmitted signal can be received
simultaneously via $M_{\text{BS}}$ RF chains associated with
different receiving vectors
$\{\boldsymbol{q}_{m}\}_{m=1}^{M_{\text{BS}}}$. Therefore the
signal received by the $m$th RF chain at the $t'$th sub-frame of
the $t$th frame can be expressed as
\begin{align}
y_{m,t',t}&=\boldsymbol{q}_{m}^T\sum_{u=1}^U\boldsymbol{H}_u\boldsymbol{p}_{t'}s_{u,t}+
w_{m,t',t}
 \label{data-model}
\end{align}
where $w_{m,t',t}$ denotes the additive white Gaussian noise
associated with the $m$th RF chain at the $t'$th sub-frame of the
$t$th frame. Our objective is to estimate the channels
$\{\boldsymbol{H}_u\}$ from the received signal $\{y_{m,t',t}\}$.
We wish to achieve a reliable channel estimation by using as few
measurements as possible. Particularly the number of pilot symbols
$T$ is assumed to be less than $U$, i.e. $T<U$, otherwise
orthogonal pilots can be employed and the joint channel estimation
problem can be decomposed as a number of single-user channel
estimation problems. In the following, we show that the received
data can be represented as a tensor and such a representation
allows a more efficient algorithm to extract the channel state
information with minimum number of measurements. Before
proceeding, we first provide a brief review of tensor and the
CANDECOMP/PARAFAC (CP) decomposition.

\begin{figure}[!t]
\centering
\includegraphics[width=8cm]{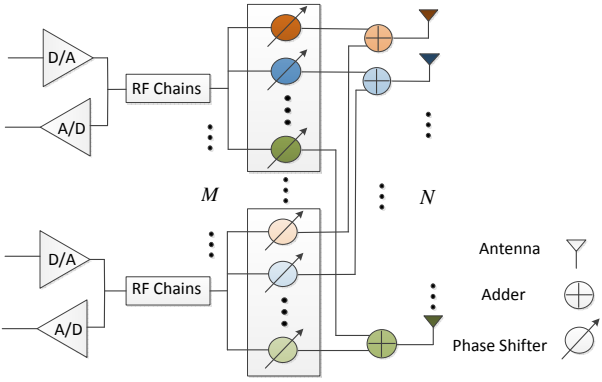}
\caption{The hybrid precoding structure for the base station and
the mobile station.} \label{fig1}
\end{figure}

\section{Preliminaries}\label{sec:tensor-overview}
We first provide a brief review on tensor and the CP
decomposition. A tensor is a generalization of a matrix to
higher-order dimensions, also known as ways or modes. Vectors and
matrices can be viewed as special cases of tensors with one and
two modes, respectively. Throughout this paper, we use symbols
$\otimes$ , $\circ$ , $\odot$ and $\ast$ to denote the Kronecker, outer,
Khatri-Rao and Hadamard product, respectively.

Let $\boldsymbol{\mathcal{X}}\in\mathbb{R}^{I_1\times
I_2\times\cdots\times I_N}$ denote an $N$th order tensor with its
$(i_1,\ldots,i_N)$th entry denoted by $\mathcal{X}_{i_1\cdots
i_N}$. Here the order $N$ of a tensor is the number of dimensions.
Fibers are higher-order analogue of matrix rows and columns. The
mode-$n$ fibers of $\boldsymbol{\mathcal{X}}$ are
$I_n$-dimensional vectors obtained by fixing every index but
$i_n$. Unfolding or matricization is an operation that turns a
tensor to a matrix. Specifically, the mode-$n$ unfolding of a
tensor $\boldsymbol{\mathcal{X}}$, denoted as
$\boldsymbol{X}_{(n)}$, arranges the mode-$n$ fibers to be the
columns of the resulting matrix. For notational convenience, we
also use the notation
$\textrm{unfold}_n(\boldsymbol{\mathcal{X}})$ to denote the
unfolding operation along the $n$-th mode. The $n$-mode product of
$\boldsymbol{\mathcal{X}}$ with a matrix
$\boldsymbol{A}\in\mathbb{R}^{J\times I_n}$ is denoted by
$\boldsymbol{\mathcal{X}}\times_n\boldsymbol{A}$ and is of size
$I_1\cdots\times I_{n-1}\times J\times I_{n+1}\times\cdots\times
I_N$, with each mode-$n$ fiber multiplied by the matrix
$\boldsymbol{A}$, i.e.
\begin{align}
\boldsymbol{\mathcal{Y}}=\boldsymbol{\mathcal{X}}\times_n\boldsymbol{A}\Leftrightarrow
\boldsymbol{Y}_{(n)}=\boldsymbol{A}\boldsymbol{X}_{(n)}
\end{align}

\begin{figure}[!t]
\centering
\includegraphics[width=9cm]{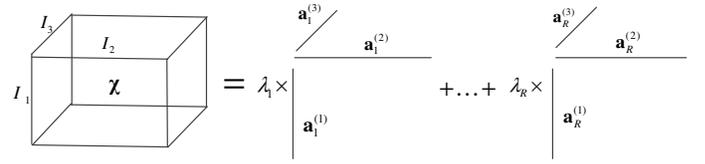}
\caption{Schematic of CP decomposition.} \label{fig-CP}
\end{figure}

The CP decomposition decomposes a tensor into a sum of rank-one
component tensors (see Fig. \ref{fig-CP}), i.e.
\begin{align}
\boldsymbol{\mathcal{X}}=
\sum\limits_{r=1}^{R}\lambda_r\boldsymbol{a}_r^{(1)}\circ\boldsymbol{a}_r^{(2)}\circ\cdots\circ\boldsymbol{a}_r^{(N)}
\end{align}
where $\boldsymbol{a}_r^{(n)}\in\mathbb{R}^{I_n}$, `$\circ$'
denotes the vector outer product, the minimum achievable $R$ is
referred to as the rank of the tensor, and
$\boldsymbol{A}^{(n)}\triangleq
[\boldsymbol{a}_{1}^{(n)}\phantom{0}\ldots\phantom{0}\boldsymbol{a}_{R}^{(n)}]\in\mathbb{R}^{I_n\times
R}$ denotes the factor matrix along the $n$-th mode. Elementwise,
we have
\begin{align}
\mathcal{X}_{i_1 i_2\cdots i_N}=\sum\limits_{r=1}^{R}\lambda_r
a_{i_1 r}^{(1)}a_{i_2 r}^{(2)}\cdots a_{i_N r}^{(N)}
\end{align}
The mode-$n$ unfolding of $\boldsymbol{\mathcal{X}}$ can be
expressed as
\begin{align}
\boldsymbol{X}_{(n)}=\boldsymbol{A}^{(n)}\boldsymbol{\Lambda}\left(\boldsymbol{A}^{(N)}
\odot\cdots\boldsymbol{A}^{(n+1)}\odot\boldsymbol{A}^{(n-1)}\odot\cdots\boldsymbol{A}^{(1)}\right)^T
\end{align}
where
$\boldsymbol{\Lambda}\triangleq\text{diag}(\lambda_1,\ldots,\lambda_R)$.
The inner product of two tensors with the same size is defined as
\begin{align}
\langle\boldsymbol{\mathcal{X}},\boldsymbol{\mathcal{Y}}\rangle =
\sum\limits_{i_1=1}^{I_1}\sum\limits_{i_2=1}^{I_2}\cdots
\sum\limits_{i_N=1}^{I_N} x_{i_1 i_2 \dots i_N} y_{i_1 i_2 \dots
i_N} \nonumber
\end{align}
The Frobenius norm of a tensor $\boldsymbol{\mathcal{X}}$ is the
square root of the inner product with itself, i.e.
\begin{align}
\|\boldsymbol{\mathcal{X}}\|_F=\langle\boldsymbol{\mathcal{X}},\boldsymbol{\mathcal{X}}\rangle^{\frac{1}{2}}
\nonumber
\end{align}

%dropped the dependence on $u$ for notations

\section{Proposed CP Decomposition-Based Channel Estimation Method}\label{sec:proposed-method}
Tensors provide a natural representation of data with multiple
modes. Note that in our data model, the received signal
$y_{m,t',t}$ has three modes which respectively stand for the RF
chain, the sub-frame and the frame. Therefore the received data
$\{y_{m,t',t}\}$ can be naturally represented by a three-mode
tensor $\boldsymbol{\mathcal{Y}}\in\mathbb{R}^{M_{\text{BS}}\times
T'\times T}$, with its $(m,t',t)$th entry given by $y_{m,t',t}$.
Combining (\ref{channel-model}) and (\ref{data-model}),
$y_{m,t',t}$ can be rewritten as
\begin{align}\label{model}
y_{m,t',t}=&\sum_{u=1}^U\sum_{j=1}^{L_u}
\alpha_{u,j}\boldsymbol{q}_{m}^T\boldsymbol{a}_{\text{BS}}(\theta_{u,j})\boldsymbol{a}_{\text{MS}}^T(\phi_{u,j})
\boldsymbol{p}_{t'}s_{u,t}+
w_{m,t',t} \nonumber\\
=&
\sum_{l=1}^L\alpha_{l}\boldsymbol{q}_{m}^T\boldsymbol{a}_{\text{BS}}(\theta_{l})\boldsymbol{a}_{\text{MS}}^T(\phi_{l})
\boldsymbol{p}_{t'}\bar{s}_{l,t}+w_{m,t',t}
\end{align}
where with a slight abuse of notation, we let
$\alpha_{l}=\alpha_{u,j}$, $\theta_{l}=\theta_{u,j}$, and
$\phi_{l}=\phi_{u,j}$, in which $l=\sum_{i=1}^{u-1}L_i+j$;
$L\triangleq\sum_{u=1}^U L_u$ denotes the total number of paths
associated with all users, and $\bar{s}_{l,t}=s_{u,t}$ if the
$l$th path comes from the $u$th user, i.e.
\begin{align}
\bar{s}_{l,t}=s_{u,t} \quad \forall l\in
\left[\sum_{i=1}^{u-1}L_{i}+1, \sum_{i=1}^{u}L_{i}\right]
\end{align}
Define
\begin{align}
\boldsymbol{Q}\triangleq &
[\boldsymbol{q}_1\phantom{0}\ldots\phantom{0}\boldsymbol{q}_{M_{\text{BS}}}]
\nonumber\\
\boldsymbol{P}\triangleq &
[\boldsymbol{p}_1\phantom{0}\ldots\phantom{0}\boldsymbol{p}_{T'}]
\nonumber
\end{align}
Since both $\boldsymbol{Q}$ and $\boldsymbol{P}$ are implemented
using analog phase shifters, their entries are of constant
modulus. Let $\boldsymbol{Y}_t\in\mathbb{R}^{M_{\text{BS}}\times
T'}$ denote a matrix obtained by fixing the index $t$ of the
tensor $\boldsymbol{\mathcal{Y}}$, we have
\begin{align}
\boldsymbol{Y}_t=&\sum_{l=1}^L \alpha_l \bar{s}_{l,t}
\boldsymbol{Q}^T\boldsymbol{a}_{\text{BS}}(\theta_{l})\boldsymbol{a}_{\text{MS}}^T(\phi_{l})
\boldsymbol{P}+ \boldsymbol{W}_t \nonumber\\
=& \sum_{l=1}^L \bar{s}_{l,t}
\boldsymbol{\tilde{a}}_{\text{BS}}(\theta_{l})\boldsymbol{\tilde{a}}_{\text{MS}}^T(\phi_{l})+
\boldsymbol{W}_t \label{Y-slice}
\end{align}
where
\begin{align}
\boldsymbol{\tilde{a}}_{\text{BS}}(\theta_{l})\triangleq&\alpha_l\boldsymbol{Q}^T\boldsymbol{a}_{\text{BS}}(\theta_{l})
\nonumber\\
\boldsymbol{\tilde{a}}_{\text{MS}}(\phi_{l})\triangleq&\boldsymbol{P}^T\boldsymbol{a}_{\text{MS}}(\phi_{l})
\nonumber
\end{align}
Since each slice of $\boldsymbol{\mathcal{Y}}$,
$\boldsymbol{Y}_t$, is a weighted sum of a common set of rank-one
outer products, the tensor $\boldsymbol{\mathcal{Y}}$ thus admits
the following CP decomposition which decomposes a tensor into a
sum of rank-one component tensors, i.e.
\begin{align}
\boldsymbol{\mathcal{Y}}=\sum_{l=1}^L
\boldsymbol{\tilde{a}}_{\text{BS}}(\theta_l)\circ\boldsymbol{\tilde{a}}_{\text{MS}}(\phi_l)\circ{\bar{\boldsymbol{s}}}_l
+\boldsymbol{\mathcal{W}} \label{CP}
\end{align}
where $\bar{\boldsymbol{s}}_l\triangleq
[\bar{s}_{l,1}\phantom{0}\cdots\phantom{0} \bar{s}_{l,T}]^T$.
Define
\begin{align}
\boldsymbol{A}_{Q}\triangleq &
[\boldsymbol{\tilde{a}}_{\text{BS}}(\theta_1)\phantom{0}\cdots\phantom{0}\boldsymbol{\tilde{a}}_{\text{BS}}(\theta_L)]
\label{AQ-definition}\\
\boldsymbol{A}_{P}\triangleq &
[\boldsymbol{\tilde{a}}_{\text{MS}}(\phi_1)\phantom{0}\cdots\phantom{0}\boldsymbol{\tilde{a}}_{\text{MS}}(\phi_L)]
\label{AP-definition}\\
{\boldsymbol{S}}_L\triangleq &
[\bar{\boldsymbol{s}}_1\phantom{0}\cdots\phantom{0}\bar{\boldsymbol{s}}_L]
\label{SL-definition}
\end{align}
Clearly,
$\{\boldsymbol{A}_{Q},\boldsymbol{A}_{P},\boldsymbol{S}_L\}$ are
factor matrices associated with a noiseless version of
$\boldsymbol{\mathcal{Y}}$. Let
\begin{align}
\boldsymbol{S}\triangleq
[\boldsymbol{s}_1\phantom{0}\cdots\phantom{0}\boldsymbol{s}_U]
\label{S-definition}
\end{align}
where
\begin{align}
{\boldsymbol{s}}_u\triangleq
[{s}_{u,1}\phantom{0}\cdots\phantom{0} {s}_{u,T}]^T
\label{su-definition}
\end{align}
then we have $\boldsymbol{S}_L=\boldsymbol{S}\boldsymbol{O}$,
where
\begin{align}
\boldsymbol{O} \triangleq \left[ {\begin{array}{*{20}{c}}
    {{\boldsymbol 1}_{{L_1}}^T}&\boldsymbol{0}& \cdots &\boldsymbol{0}\\
    \boldsymbol{0}&{{\boldsymbol 1}_{{L_2}}^T}& \cdots &\boldsymbol{0}\\
    \vdots & \vdots & \ddots & \vdots \\
    \boldsymbol{0}&\boldsymbol{0}& \cdots &{{\boldsymbol 1}_{{L_U}}^T}
    \end{array}} \right] \label{O}
\end{align}
where $\boldsymbol{1}_{l}$ denotes an $l$-dimensional column
vector with all entries equal to one. Equation (\ref{CP}) suggests
that an estimate of the mmWave channels $\{\boldsymbol{H}_u\}$ can
be obtained by performing a CP decomposition of the tensor
$\boldsymbol{\mathcal{Y}}$.

\subsection{CP Decomposition}\label{subsec:CP-decomposition}
Given that the number of total paths, $L$, is known \emph{a
priori}\footnote{This could be the case if there is only a direct
line-of-sight path between each user and the BS, in which case we
have $L=U$.}, the CP decomposition can be accomplished by solving
the following optimization problem
\begin{align}
\min_{\boldsymbol{A}_{Q},\boldsymbol{A}_{P},\boldsymbol{S}_L}\quad
\|\boldsymbol{\mathcal{Y}}-\sum_{l=1}^L
\boldsymbol{\tilde{a}}_{\text{BS}}(\theta_l)
\circ\boldsymbol{\tilde{a}}_{\text{MS}}(\phi_l)\circ{\bar{\boldsymbol{s}}}_l\|_{F}^2
\label{opt-1}
\end{align}
The above optimization can be efficiently solved by an alternating
least squares (ALS) procedure which iteratively minimizes the data
fitting error with respect to the three factor matrices:
\begin{align}
\boldsymbol{A}_{Q}^{(t+1)}=&\arg\min_{\boldsymbol{A}_{Q}}
\left\|\boldsymbol{Y}_{(1)}^T-(\boldsymbol{S}^{(t)}_L\odot\boldsymbol{A}_{P}^{(t)})\boldsymbol{A}_{Q}^T
\right\|_F^2 \label{AQ-update} \\
\boldsymbol{A}_{P}^{(t+1)}=&\arg\min_{\boldsymbol{A}_{P}}
\left\|\boldsymbol{Y}_{(2)}^T-(\boldsymbol{S}^{(t)}_L\odot\boldsymbol{A}_{Q}^{(t+1)})\boldsymbol{A}_{P}^T
\right\|_F^2 \label{AP-update} \\
{\boldsymbol{S}}^{(t+1)}_L=&\arg\min_{{\boldsymbol{S}_L}}
\left\|\boldsymbol{Y}_{(3)}^T-(\boldsymbol{A}_{P}^{(t+1)}
\odot\boldsymbol{A}_{Q}^{(t+1)})\boldsymbol{S}_L^T \right\|_F^2
\label{S-update}
\end{align}

For the general case where the total number of paths $L$ is
unknown \emph{a priori}, more sophisticated CP decomposition
techniques can be used to jointly estimate the model order and the
factor matrices. Since $L$ is usually small relative to the
dimensions of the tensor, the factorization (\ref{CP}) implies
that the tensor $\boldsymbol{\mathcal{Y}}$ has a low-rank
structure. Hence the CP decomposition can be cast as a rank
minimization problem as
\begin{align}
\min_{\boldsymbol{\mathcal{X}}}&\quad
\text{rank}(\boldsymbol{\mathcal{X}}) \nonumber\\
\text{s.t.}&\quad \|\boldsymbol{\mathcal{Y}}-
\boldsymbol{\mathcal{X}}\|_F^2\leq\varepsilon \label{opt-2}
\end{align}
where $\varepsilon$ is an error tolerance parameter related to
noise statistics. Note that the CP rank is the minimum number of
rank-one tensor components required to represent the tensor. Thus
the search for a low rank $\boldsymbol{\mathcal{X}}$ can be
converted to the optimization of its associated factor matrices.
Let
\begin{align}
\boldsymbol{\mathcal{X}}=\sum\limits_{k=1}^{K}\boldsymbol{a}_k\circ
\boldsymbol{b}_k\circ\boldsymbol{c}_k
\end{align}
where $K\gg L$ denotes an upper bound of the total number of
paths, and
\begin{align}
\boldsymbol{A}\triangleq &
[\boldsymbol{a}_1\phantom{0}\ldots\phantom{0}\boldsymbol{a}_{K}]
\nonumber\\
\boldsymbol{B}\triangleq &
[\boldsymbol{b}_1\phantom{0}\ldots\phantom{0}\boldsymbol{b}_{K}]
\nonumber\\
\boldsymbol{C}\triangleq &
[\boldsymbol{c}_1\phantom{0}\ldots\phantom{0}\boldsymbol{c}_{K}]
\nonumber
\end{align}
The optimization (\ref{opt-2}) can be re-expressed as
\begin{align}
\min_{\boldsymbol{A},\boldsymbol{B},\boldsymbol{C}}&\quad
\|\boldsymbol{z}\|_0 \nonumber\\
\text{s.t.} &\quad \|\boldsymbol{\mathcal{Y}}-
\boldsymbol{\mathcal{X}}\|_F^2\leq\varepsilon
\nonumber\\
&\quad
\boldsymbol{\mathcal{X}}=\sum\limits_{k=1}^{K}\boldsymbol{a}_k\circ
\boldsymbol{b}_k\circ\boldsymbol{c}_k \label{opt-3}
\end{align}
where $\boldsymbol{z}$ is a $K$-dimensional vector with its $k$th
entry given by
\begin{align}
z_{k}\triangleq\|\boldsymbol{a}_k\circ
\boldsymbol{b}_k\circ\boldsymbol{c}_k\|_{F}
\end{align}
We see that $\|\boldsymbol{z}\|_{0}$ equals to the number of
nonzero rank-one tensor components. Therefore minimizing the
$\ell_0$-norm of $\boldsymbol{z}$ is equivalent to minimizing the
rank of the tensor $\boldsymbol{\mathcal{X}}$.

The optimization (\ref{opt-3}) is an NP-hard problem.
Nevertheless, alternative sparsity-promoting functions such as
$\ell_1$-norm can be used to replace $\ell_0$-norm to find a
sparse solution of $\boldsymbol{z}$ more efficiently. In this
paper, we use $\|\cdot \|_{2/3}$ as the relaxation of
$\|\cdot\|_0$. From \cite{BazerqueMateos13}, we know that
$(\|\boldsymbol{z}\|_{2/3})^{3/2}=\|\boldsymbol{\mathcal{X}}\|_{\ast}$,
where
\begin{align}
\|\boldsymbol{\mathcal{X}}\|_{\ast}\triangleq
\text{tr}({\boldsymbol{A}}{{\boldsymbol{A}}^H}) +
\text{tr}({\boldsymbol{B}}{{\boldsymbol{B}}^H}) +
\text{tr}({\boldsymbol{C}}{{\boldsymbol{C}}^H}) \nonumber
\end{align}
Thus (\ref{opt-3}) can be relaxed as the following optimization
problem
\begin{align}\label{opt-1}
\mathop {\min
}\limits_{\boldsymbol{A},\boldsymbol{B},\boldsymbol{C}}\quad
&\left\| {\boldsymbol{\mathcal Y} -\boldsymbol{\mathcal X}}
\right\|_F^2 +
\mu {\left\| \boldsymbol{\mathcal{X}} \right\|_*} \nonumber\\
\text{s.t.}\quad &
\boldsymbol{\mathcal{X}}=\sum\limits_{k=1}^{K}\boldsymbol{a}_k\circ
\boldsymbol{b}_k\circ\boldsymbol{c}_k
\end{align}
where $\mu$ is a regularization parameter whose choice will be
discussed later in this paper. Again, the above optimization can
be efficiently solved by an alternating least squares (ALS)
procedure which iteratively minimizes (\ref{opt-1}) with respect
to the three factor matrices:
\begin{equation}
{{\boldsymbol{A}}^{(t + 1)}} = \mathop {\arg \min
}\limits_{\boldsymbol{A}} \left\| {\left[ {\begin{array}{*{20}{c}}
{\boldsymbol{Y}_{(1)}^T}\\
\bf 0
\end{array}} \right] - \left[ {\begin{array}{*{20}{c}}
{{{\boldsymbol{C}}^{(t)}} \odot {{\boldsymbol{B}}^{(t)}}}\\
{\sqrt \mu  {\boldsymbol{I}}}
\end{array}} \right]{{\boldsymbol{A}}^T}} \right\|_F^2
\label{updateA}
\end{equation}

\begin{equation}
{{\boldsymbol{B}}^{(t + 1)}} = \mathop {\arg \min
}\limits_{\boldsymbol{B}} \left\| {\left[ {\begin{array}{*{20}{c}}
{\boldsymbol{Y}_{(2)}^T}\\
\bf 0
\end{array}} \right] - \left[ {\begin{array}{*{20}{c}}
{{{\boldsymbol{C}}^{(t)}} \odot {{\boldsymbol{A}}^{(t+1)}}}\\
{\sqrt \mu  {\boldsymbol{I}}}
\end{array}} \right]{{\boldsymbol{B}}^T}} \right\|_F^2 \label{updateB}
\end{equation}

\begin{equation}
{{\boldsymbol{C}}^{(t + 1)}} = \mathop {\arg \min
}\limits_{\boldsymbol{C}} \left\| {\left[ {\begin{array}{*{20}{c}}
{\boldsymbol{Y}_{(3)}^T}\\
\bf 0
\end{array}} \right] - \left[ {\begin{array}{*{20}{c}}
{{{\boldsymbol{B}}^{(t+1)}} \odot {{\boldsymbol{A}}^{(t+1)}}}\\
{\sqrt \mu  {\boldsymbol{I}}}
\end{array}} \right]{{\boldsymbol{C}}^T}} \right\|_F^2 \label{updateC}
\end{equation}
We can repeat the above iterations until the difference between
estimated factor matrices of successive iterations is negligible,
i.e. smaller than a pre-specified tolerance value. The rank of the
tensor can be estimated by removing those negligible rank-one
tensor components. Note that during the decomposition, we do not
need to impose a specific structure on the estimates of the factor
matrices since the CP decomposition is unique under very mild
conditions.

\subsection{Channel Estimation}
We now discuss how to estimate the mmWave channel based on the
estimated factor matrices
$\{\boldsymbol{\hat{A}}_{Q},\boldsymbol{\hat{A}}_{P},\boldsymbol{\hat{S}}_L\}$.
As to be shown in
(\ref{factor-matrix-relation1})--(\ref{factor-matrix-relation4}),
under a mild condition, the estimated factor matrices and the true
factor matrices are related as follows
\begin{align}
\label{AQ}
\boldsymbol{\hat{A}}_{Q}=&\boldsymbol{A}_{Q}\boldsymbol{\Lambda}_1\boldsymbol{\Pi}
+\boldsymbol{E}_1 \\
\label{AP}
\boldsymbol{\hat{A}}_{P}=&\boldsymbol{A}_{P}\boldsymbol{\Lambda}_2\boldsymbol{\Pi}
+\boldsymbol{E}_2 \\
\label{Sprime}
\boldsymbol{\hat{S}}_L=&\boldsymbol{S}_L\boldsymbol{\Lambda}_3\boldsymbol{\Pi}
+\boldsymbol{E}_3
\end{align}
where $\boldsymbol{\Lambda}_3$ is a nonsingular diagonal matrix,
$\boldsymbol{\Lambda}_1$ and $\boldsymbol{\Lambda}_2$ are
nonsingular block diagonal matrices compatible with the block
structure of $\boldsymbol{A}_{Q}$ and $\boldsymbol{A}_{P}$,
respectively, and we have
$\boldsymbol{\Lambda}_1\boldsymbol{\Lambda}_3\boldsymbol{\Lambda}_2^T=\boldsymbol{I}$;
$\boldsymbol{\Pi}$
is a permutation matrix, $\boldsymbol{E}_1$, $\boldsymbol{E}_2$,
and $\boldsymbol{E}_3$ denote the estimation errors associated
with the three estimated factor matrices, respectively. Note that
both $\boldsymbol{A}_{Q}$ and $\boldsymbol{A}_{P}$ can be
partitioned into $U$ blocks with each block consisting of column
vectors associated with each user, i.e.
\begin{align}
\boldsymbol{A}_{Q}=&[\boldsymbol{A}_{Q,1}\phantom{0}\boldsymbol{A}_{Q,2}\phantom{0}\ldots\phantom{0}\boldsymbol{A}_{Q,U}]
\\
\boldsymbol{A}_{P}=&[\boldsymbol{A}_{P,1}\phantom{0}\boldsymbol{A}_{P,2}\phantom{0}\ldots\phantom{0}\boldsymbol{A}_{P,U}]
\end{align}
in which
\begin{align}
\boldsymbol{A}_{Q,u}\triangleq &
[\boldsymbol{\tilde{a}}_{\text{BS}}(\theta_{u,1})\phantom{0}\ldots\phantom{0}
\boldsymbol{\tilde{a}}_{\text{BS}}(\theta_{u,L_u})] \label{AQu}
\\
\boldsymbol{A}_{P,u}\triangleq &
[\boldsymbol{\tilde{a}}_{\text{MS}}(\phi_{u,1})\phantom{0}\ldots\phantom{0}\boldsymbol{\tilde{a}}_{\text{MS}}(\phi_{u,L_u})]
\label{APu}
\end{align}
The block-diagonal structure of $\boldsymbol{\Lambda}_1$ and
$\boldsymbol{\Lambda}_2$ is compatible with the block structure of
$\boldsymbol{A}_{Q}$ and $\boldsymbol{A}_{P}$. Thus we have
\begin{align}
\boldsymbol{\Lambda}_1=&\text{diag}(\boldsymbol{\Lambda}_{1}^{(1)},\ldots,\boldsymbol{\Lambda}_{1}^{(U)})
\\
\boldsymbol{\Lambda}_2=&\text{diag}(\boldsymbol{\Lambda}_{2}^{(1)},\ldots,\boldsymbol{\Lambda}_{2}^{(U)})
\end{align}
and
\begin{align}
\boldsymbol{A}_{Q}\boldsymbol{\Lambda}_1=[\boldsymbol{A}_{Q,1}\boldsymbol{\Lambda}_{1}^{(1)}
\phantom{0}\ldots\phantom{0}\boldsymbol{A}_{Q,U}\boldsymbol{\Lambda}_{1}^{(U)}]
\\
\boldsymbol{A}_{P}\boldsymbol{\Lambda}_2=[\boldsymbol{A}_{P,1}\boldsymbol{\Lambda}_{2}^{(1)}
\phantom{0}\ldots\phantom{0}\boldsymbol{A}_{P,U}\boldsymbol{\Lambda}_{2}^{(U)}]
\end{align}
The diagonal matrix $\boldsymbol{\Lambda}_3$ can also be
partitioned according to the structure of $\boldsymbol{\Lambda}_1$
and $\boldsymbol{\Lambda}_2$:
\begin{align}
\boldsymbol{\Lambda}_3=&\text{diag}(\boldsymbol{\Lambda}_{3}^{(1)},\ldots,\boldsymbol{\Lambda}_{3}^{(U)})
\end{align}
From
$\boldsymbol{\Lambda}_1\boldsymbol{\Lambda}_3\boldsymbol{\Lambda}_2^T=\boldsymbol{I}$,
we can readily arrive at
\begin{align}
\boldsymbol{\Lambda}_1^{(u)}\boldsymbol{\Lambda}_3^{(u)}(\boldsymbol{\Lambda}_2^{(u)})^T=\boldsymbol{I}
\quad \forall u=1,\ldots,U
\end{align}

To estimate the channel, we first estimate the number of paths
associated with each user, the diagonal matrix
${\boldsymbol\Lambda}_3$ and the permutation matrix $\boldsymbol
\Pi$ from (\ref{Sprime}). Suppose there are no estimation errors,
each column of $\boldsymbol{\hat{S}}_L$ is a scaled version of a
training sequence associated with an unknown user. Since the
training sequences of all users are known \emph{a priori}, a
simple correlation-based matching method can be used to determine
the unknown scaling factor and the permutation ambiguity for each
column of $\boldsymbol{\hat{S}}_L$, based on which the number of
paths associated with each user, the diagonal matrix
${\boldsymbol\Lambda}_3$ and the permutation matrix
$\boldsymbol{\Pi}$ can be readily obtained.

Suppose the diagonal matrix ${\boldsymbol\Lambda}_3$ and the
permutation matrix $\boldsymbol{\Pi}$ are perfectly recovered. The
permutation ambiguity for the estimated factor matrices
$\boldsymbol{\hat{A}}_{Q}$ and $\boldsymbol{\hat{A}}_{P}$ can be
removed using the estimated permutation matrix. Thus we have
\begin{align}
\boldsymbol{\hat{A}}_{Q}=&\boldsymbol{A}_{Q}\boldsymbol{\Lambda}_1+\boldsymbol{E}_1
\\
\boldsymbol{\hat{A}}_{P}=&\boldsymbol{A}_{P}\boldsymbol{\Lambda}_2+\boldsymbol{E}_2
\end{align}
Given $\boldsymbol{\hat{A}}_{Q}$, $\boldsymbol{\hat{A}}_{P}$ and
${\boldsymbol\Lambda}_3$, the $u$th user's channel matrix
$\boldsymbol{H}_u$ can be estimated from
$\boldsymbol{\hat{A}}_{Q,u}\boldsymbol{\Lambda}_3^{(u)}\boldsymbol{\hat{A}}_{P,u}$
since we have
\begin{align}
\boldsymbol{\hat{A}}_{Q,u}\boldsymbol{\Lambda}_3^{(u)}\boldsymbol{\hat{A}}_{P,u}=
&\boldsymbol{A}_{Q,u}\boldsymbol{\Lambda}_{1}^{(u)}\boldsymbol{\Lambda}_3^{(u)}
(\boldsymbol{\Lambda}_{2}^{(u)})^T\boldsymbol{A}_{P,u}^T +
\boldsymbol{E}
\nonumber\\
=&\boldsymbol{A}_{Q,u}\boldsymbol{A}_{P,u}^T +
\boldsymbol{E}\nonumber\\
=&\sum_{l=1}^{L_u}\boldsymbol{\tilde{a}}_{\text{BS}}(\theta_{u,l})\boldsymbol{\tilde{a}}_{\text{MS}}(\phi_{u,l})^T
+ \boldsymbol{E}
\nonumber\\
=&\boldsymbol{Q}^T\sum_{l=1}^{L_u}\alpha_{u,l}\boldsymbol{a}_{\text{BS}}(\theta_{u,l})
\boldsymbol{a}_{\text{MS}}(\phi_{u,l})^T \boldsymbol{P} +
\boldsymbol{E}
\nonumber\\
=&\boldsymbol{Q}^T\boldsymbol{H}_u\boldsymbol{P} + \boldsymbol{E}
\label{Hu-estimation}
\end{align}
where $\boldsymbol{E}$ denotes the estimation error caused by
$\boldsymbol{E}_1$ and $\boldsymbol{E}_2$. We see that the joint
multiuser channel estimation has been decoupled into $U$
single-user channel estimation problems via the CP factorization.
In the following section, we will show that the uniqueness of the
decomposition can be guaranteed even when $T\ll U$. This enables a
significant training overhead reduction since traditional
estimation methods rely on the use of orthogonal pilot sequences
(which requires $T=U$) to decouple the multiuser channel
estimation problem into a set of single-user channel estimation
problems. Let
$\boldsymbol{z}_u\triangleq\text{vec}(\boldsymbol{\hat{A}}_{Q,u}\boldsymbol{\Lambda}_3^{(u)}\boldsymbol{\hat{A}}_{P,u})$.
We have
\begin{align}\label{Hu-estimation2}
\boldsymbol{z}_u=(\boldsymbol{P}^T\otimes\boldsymbol{Q}^T)\boldsymbol{\tilde{H}}_u\boldsymbol{\alpha}_u+\boldsymbol{e}
\end{align}
where $\boldsymbol{\alpha}_u\triangleq
[\alpha_{u,1}\phantom{0}\ldots\phantom{0}\alpha_{u,L_u}]$, and
\begin{align}
\boldsymbol{\tilde{H}}_u\triangleq
[\boldsymbol{a}_{\text{MS}}(\phi_{u,1})\otimes
\boldsymbol{a}_{\text{BS}}(\theta_{u,1})\phantom{0}\ldots\phantom{0}
\boldsymbol{a}_{\text{MS}}(\phi_{u,L})\otimes
\boldsymbol{a}_{\text{BS}}(\theta_{u,L_u})]
\end{align}
The estimation of $\boldsymbol{\tilde{H}}_u$ can be cast as a
compressed sensing problem by discretizing the continuous
parameter space into an $N_1\times N_2$ two dimensional grid with
each grid point given by $\{\bar{\theta}_{i},\bar{\phi}_j\}$ for
$i=1,\ldots, N_1$ and $j=1,\ldots,N_2$ and assuming that
$\{\phi_{u,l},\theta_{u,l}\}_{l=1}^{L_u}$ lie on the grid. Thus
(\ref{Hu-estimation2}) can be re-expressed as
\begin{align}\label{Hu-estimation3}
\boldsymbol{z}_u=(\boldsymbol{P}^T\otimes\boldsymbol{Q}^T)\boldsymbol{\bar{\Sigma}}
{\bar{\boldsymbol{\alpha}}}_u+\boldsymbol{e}
\end{align}
where $\boldsymbol{\bar{\Sigma}}$ is an overcomplete dictionary
consisting of $N_1\times N_2$ columns, with its $((i-1)N_1+j)$th
column given by $\boldsymbol{a}_{\text{MS}}(\bar{\phi}_i)\otimes
\boldsymbol{a}_{\text{BS}}(\bar{\theta}_j)$,
${\bar{\boldsymbol{\alpha}}}_u\in\mathbb{C}^{N_1N_2\times 1}$ is a
sparse vector obtained by augmenting $\boldsymbol{\alpha}_u$ with
zero elements.

\section{Uniqueness}\label{sec:uniqueness-analysis}
In this section, we discuss under what conditions the uniqueness
of the CP decomposition and, in turn, the channel estimation  can
be guaranteed.

\subsection{Uniqueness for the Single-Path Geometric Model}
We first consider the special case where there is a direct
line-of-sight path between the BS and each user, in which case we
have $L=U$ and $\boldsymbol{S}_L=\boldsymbol{S}$ (recalling
$\boldsymbol{S}_L=\boldsymbol{S}\boldsymbol{O}$). It is well known
that essential uniqueness of the CP decomposition can be
guaranteed by the Kruskal's condition \cite{Kruskal77}. Let
$k_{\boldsymbol{A}}$ denote the k-rank of a matrix
$\boldsymbol{A}$, which is defined as the largest value of
$k_{\boldsymbol{A}}$ such that every subset of
$k_{\boldsymbol{A}}$ columns of the matrix $\boldsymbol{A}$ is
linearly independent. Kruskal showed that a CP decomposition
$(\boldsymbol{A},\boldsymbol{B},\boldsymbol{C})$ of a third-order
tensor is essentially unique if \cite{Kruskal77}
\begin{align}
k_{\boldsymbol{A}}+k_{\boldsymbol{B}}+k_{\boldsymbol{C}}\geq 2R+2
\label{Kruskals-condition}
\end{align}
where $\boldsymbol{A},\boldsymbol{B},\boldsymbol{C}$ are factor
matrices, and $R$ denotes the CP rank. More formally, we have the
following theorem.

\newtheorem{theorem}{Theorem}
\begin{theorem} \label{theorem1}
Let $(\boldsymbol{A},\boldsymbol{B},\boldsymbol{C})$ be a CP
solution which decomposes a three-mode tensor
$\boldsymbol{\mathcal{X}}$ into $R$ rank-one arrays. Suppose
Kruskal's condition (\ref{Kruskals-condition}) holds and there is
an alternative CP solution
$(\boldsymbol{\bar{A}},\boldsymbol{\bar{B}},\boldsymbol{\bar{C}})$
which also decomposes $\boldsymbol{\mathcal{X}}$ into $R$ rank-one
arrays. Then we have
$\boldsymbol{\bar{A}}=\boldsymbol{A}\boldsymbol{\Pi}\boldsymbol{\Lambda}_a$,
$\boldsymbol{\bar{B}}=\boldsymbol{B}\boldsymbol{\Pi}\boldsymbol{\Lambda}_b$,
and
$\boldsymbol{\bar{C}}=\boldsymbol{C}\boldsymbol{\Pi}\boldsymbol{\Lambda}_c$,
where $\boldsymbol{\Pi}$ is a unique permutation matrix and
$\boldsymbol{\Lambda}_a$, $\boldsymbol{\Lambda}_b$, and
$\boldsymbol{\Lambda}_c$ are unique diagonal matrices such that
$\boldsymbol{\Lambda}_a\boldsymbol{\Lambda}_b\boldsymbol{\Lambda}_c=\boldsymbol{I}$.
\end{theorem}
\begin{proof}
Please refer to \cite{StegemanSidiropoulos07}.
\end{proof}

From Theorem \ref{theorem1}, we know that if the following
condition holds
\begin{align}
k_{\boldsymbol{A}_Q}+k_{\boldsymbol{A}_P}+k_{\boldsymbol{S}}\geq
2U+2 \label{K-condition}
\end{align}
then the CP decomposition of $\boldsymbol{\mathcal{Y}}$ is unique
and in the noiseless case, we can ensure that the factor matrices
can be estimated up to a permutation and scaling ambiguity, i.e.
$\boldsymbol{\hat{A}}_{Q}=\boldsymbol{A}_{Q}\boldsymbol{\Pi}\boldsymbol{\Lambda}_1$,
$\boldsymbol{\hat{A}}_{P}=\boldsymbol{A}_{P}\boldsymbol{\Pi}\boldsymbol{\Lambda}_2$,
and
$\boldsymbol{\hat{S}}=\boldsymbol{S}\boldsymbol{\Pi}\boldsymbol{\Lambda}_3$,
with
$\boldsymbol{\Lambda}_1\boldsymbol{\Lambda}_2\boldsymbol{\Lambda}_3=\boldsymbol{I}$.

We now discuss how to design the beamforming matrix
$\boldsymbol{P}\in\mathbb{C}^{N_{\text{MS}}\times T'}$, the
combining matrix $\boldsymbol{Q}\in\mathbb{C}^{N_{\text{BS}}\times
M_{\text{BS}}}$, and the pilot symbol matrix
$\boldsymbol{S}\in\mathbb{C}^{T\times U}$ such that the Kruskal's
condition (\ref{K-condition}) can be met. Note that
$\boldsymbol{A}_{Q}=\boldsymbol{Q}^T\boldsymbol{A}_{\text{BS}}$,
where $\boldsymbol{A}_{\text{BS}}$ is a Vandermonte matrix whose
k-rank is equivalent to the number of columns, $U$, when the
angles of arrival $\{\theta_u\}$ are distinct. The k-rank of
$\boldsymbol{A}_{Q}$, therefore, is no greater than $U$, i.e.
$k_{\boldsymbol{A}_{Q}}\leq U$. The problem now becomes whether we
can design a combining matrix $\boldsymbol{Q}$ such that
$k_{\boldsymbol{A}_{Q}}$ achieves its upper bound $U$. We will
show that the answer is affirmative for a randomly generated
$\boldsymbol{Q}$ with i.i.d. entries. Specifically, we assume each
entry of $\boldsymbol{Q}$ is chosen uniformly from a unit circle
scaled by a constant $1/N_{\text{BS}}$, i.e.
$q_{m,n}=(1/N_{\text{BS}})e^{j\vartheta_{m,n}}$, where
$\vartheta_{m,n}\in [-\pi,\pi]$ follows a uniform distribution.
Let $a_{m,i}\triangleq
\boldsymbol{q}_m^T\boldsymbol{a}_{\text{BS}}(\theta_i)$ denote the
$(m,i)$th entry of $\boldsymbol{A}_{Q}$. It can be readily
verified that ${\mathbb E}[a_{m,i}]=0,\forall m,i$ and
\begin{equation}
{\mathbb E}[a_{m,i}a_{n,j}^{\ast}]=\begin{cases}0 & m\neq n \\
\frac{1}{N_{\text{BS}}^2}\boldsymbol{a}_{\text{BS}}^H(\theta_i)\boldsymbol{a}_{\text{BS}}(\theta_j)
& m=n
\end{cases}
\end{equation}
When the number of antennas at the BS is sufficiently large, the
steering vectors $\{\boldsymbol{a}_{\text{BS}}(\theta_i)\}$ become
mutually quasi-orthogonal, i.e.
$\boldsymbol{a}_{\text{BS}}^H(\theta_i)\boldsymbol{a}_{\text{BS}}(\theta_j)\rightarrow
\delta(\theta_i-\theta_j)$, which implies that the entries of
$\boldsymbol{A}_{Q}$ are uncorrelated with each other. On the
other hand, according to the central limit theorem, we know that
each entry $a_{m,i}$ approximately follows a Gaussian
distribution. Therefore entries of $\boldsymbol{A}_{Q}$ can be
considered as i.i.d. Gaussian variables, and $\boldsymbol{A}_{Q}$
is full column rank with probability one. Thus we can reach that
the k-rank of $\boldsymbol{A}_{Q}$ is equivalent to $U$ with
probability one.

%then we have $k_{\boldsymbol{S}}=2$

Following a similar derivation, we can arrive at the following
conclusion: if each entry of the beamforming matrix
$\boldsymbol{P}$ is chosen uniformly from a unit circle scaled by
a constant $1/N_{\text{MS}}$, then the k-rank of
$\boldsymbol{A}_{P}$ is equivalent to $U$ with probability one.
Thus we can guarantee that the Kruskal's condition
(\ref{K-condition}) is met with probability one as long as
$k_{\boldsymbol{S}}\geq 2$, i.e. any two columns of
$\boldsymbol{S}$ are linearly independent. For the single path
geometric model, $\boldsymbol{S}$ consists of $U$ columns, with
the $u$th column constructed by pilot symbols of the $u$th user.
Therefore the condition $k_{\boldsymbol{S}}\geq 2$ can be ensured
provided that $T\geq 2$, and pilot symbol vectors of users are
mutually independent. Specifically, we can design the pilot
symbols by minimizing the mutual coherence of $\boldsymbol{S}$,
i.e.
\begin{align}\label{codebook}
\min_{\boldsymbol{S}}\quad \mu(\boldsymbol{S})
\end{align}
where
\begin{align}
\mu(\boldsymbol{S})\triangleq\max_{i\neq j}
\left|\frac{\langle\boldsymbol{s}_i,\boldsymbol{s}_j\rangle}{\|\boldsymbol{s}_i\|\|\boldsymbol{s}_j\|}\right|
\nonumber
\end{align}
The solution of above problem can be found in
\cite{ConwayHardin96,StrohmerThomas03}. For the case
$k_{\boldsymbol{A}_{Q}}=U$ and $k_{\boldsymbol{A}_{P}}=U$, the
Kruskal's condition can be met by choosing the length of the pilot
sequence equal to two, i.e. $T=2$, irrespective of the value of
$U$. This allows a considerable training overhead reduction,
particularly when $U$ is large. Note that besides random coding,
the beamforming and combining matrices $\boldsymbol{P}$ and
$\boldsymbol{Q}$ can also be devised to form a certain number of
transmit/receive beams. The k-rank of the resulting matrices
$\boldsymbol{A}_{P}$ and $\boldsymbol{A}_{Q}$ may also achieve the
upper bound $U$.

\subsection{Uniqueness for the General Geometric Model}
For the general geometric model where there are more than one path
between each user and the BS, the Kruskal's condition becomes
\begin{align}
k_{\boldsymbol{A}_Q}+k_{\boldsymbol{A}_P}+k_{\boldsymbol{S}_L}\geq
2L+2
\end{align}
Since the k-rank of
$\boldsymbol{A}_Q\in\mathbb{C}^{M_{\text{BS}}\times L}$ and
$\boldsymbol{A}_P\in\mathbb{C}^{T'\times L}$ is at most equal to
$L$, we need $k_{\boldsymbol{S}_L}\geq 2$ to satisfy the above
Kruskal's condition. However, for the general geometric model, the
k-rank of $\boldsymbol{S}_L$ is always equal to one because
multiple column vectors associated with a common user are linearly
dependent. Thus the Kruskal's condition can never be satisfied in
this case. Nevertheless, this does not mean that the uniqueness of
the CP decomposition does not hold for the general geometric
model. In fact, considering the special form of the decomposition
(\ref{CP}), the uniqueness can be guaranteed under a less
restrictive condition.

We first write (\ref{CP}) as follows
\begin{align}
\boldsymbol{\mathcal{Y}}=&\sum_{u=1}^U\sum_{l=1}^{L_u}
\boldsymbol{\tilde{a}}_{\text{BS}}(\theta_{u,l})\circ\boldsymbol{\tilde{a}}_{\text{MS}}(\phi_{u,l})
\circ\boldsymbol{{s}}_{u}
+\boldsymbol{\mathcal{W}} \nonumber\\
=&\sum_{u=1}^U
(\boldsymbol{A}_{Q_u}\boldsymbol{A}_{P_u}^T)\circ\boldsymbol{{s}}_{u}+\boldsymbol{\mathcal{W}}
\end{align}
where $\boldsymbol{A}_{Q_u}$ and $\boldsymbol{A}_{P_u}$ are
defined in (\ref{AQu}) and (\ref{APu}), respectively, and
$\boldsymbol{{s}}_{u}$ is defined in (\ref{su-definition}). We see
that the tensor $\boldsymbol{\mathcal{Y}}$ can be expressed as a
sum of matrix-vector outer products, more specifically, a sum of
rank-$(L_u,L_u,1)$ terms since $\boldsymbol{A}_{Q_u}$ and
$\boldsymbol{A}_{P_u}$ are both rank-$L_u$. For this block term
decomposition, we have the following generalized version of the
Kruskal's condition.

Before proceeding, we define $\boldsymbol{A}\triangleq
[\boldsymbol{A}_1\phantom{0}\ldots\phantom{0}\boldsymbol{A}_R]$,
$\boldsymbol{B}\triangleq
[\boldsymbol{B}_1\phantom{0}\ldots\phantom{0}\boldsymbol{B}_R]$,
and $\boldsymbol{C}\triangleq
[\boldsymbol{c}_1\phantom{0}\ldots\phantom{0}\boldsymbol{c}_R]$,
and generalize the k-rank concept to the above partitioned
matrices. Specifically, the $k'$-rank of a partitioned matrix
$\boldsymbol{A}$, denoted by $k'_{\boldsymbol{A}}$, is the maximal
number $r$ such that any set of $r$ submatrices of
$\boldsymbol{A}$ yields a set of linearly independent columns.

We have the following theorem.
\begin{theorem} \label{theorem2}
Let $(\boldsymbol{A},\boldsymbol{B},\boldsymbol{C})$ represent a
decomposition of $\boldsymbol{\mathcal{X}}\in\mathbb{C}^{M\times
N\times K}$ in rank-$(L_r,L_r,1)$ terms, i.e.
\begin{align}
\boldsymbol{\mathcal{X}}=\sum_{r=1}^R(\boldsymbol{A}_r\boldsymbol{B}_r^T)\circ\boldsymbol{c}_r
\nonumber
\end{align}
We assume $M\geq\max_r L_r$, $N\geq\max_r L_r$,
$\text{rank}(\boldsymbol{A}_r)=L_r$, and
$\text{rank}(\boldsymbol{B}_r)=L_r$. Suppose the following
conditions
\begin{align}
MN\geq \sum_{r=1}^R L_r^2
\end{align}
\begin{align}
k'_{\boldsymbol{A}}+k'_{\boldsymbol{B}}+k_{\boldsymbol{C}}\geq
2R+2
\end{align}
hold and we have an alternative decomposition of
$\boldsymbol{\mathcal{X}}$, represented by
$(\boldsymbol{\bar{A}},\boldsymbol{\bar{B}},\boldsymbol{\bar{C}})$,
with $k'_{\boldsymbol{\bar{A}}}$ and $k'_{\boldsymbol{\bar{B}}}$
maximal under the given dimensionality constraints. Then
$(\boldsymbol{A},\boldsymbol{B},\boldsymbol{C})$ and
$(\boldsymbol{\bar{A}},\boldsymbol{\bar{B}},\boldsymbol{\bar{C}})$
are essentially equal, i.e.
$\boldsymbol{\bar{A}}=\boldsymbol{A}\boldsymbol{\Pi}\boldsymbol{\Lambda}_a$,
$\boldsymbol{\bar{B}}=\boldsymbol{B}\boldsymbol{\Pi}\boldsymbol{\Lambda}_b$
and
$\boldsymbol{\bar{C}}=\boldsymbol{C}\boldsymbol{\Pi}_c\boldsymbol{\Lambda}_c$,
in which $\boldsymbol{\Pi}$ is a block permutation matrix whose
block structure is consistent with that of $\boldsymbol{A}$ and
$\boldsymbol{B}$, $\boldsymbol{\Pi}_c$ is permutation matrix whose
permutation pattern is the same as that of $\boldsymbol{\Pi}$,
$\boldsymbol{\Lambda}_a$ and $\boldsymbol{\Lambda}_b$ are
nonsingular block-diagonal matrices, compatible with the block
structure of $\boldsymbol{A}$ and $\boldsymbol{B}$, and
$\boldsymbol{\Lambda}_c$ is a nonsingular diagonal matrix. Also,
let $\boldsymbol{\Lambda}_{a,r}$ and $\boldsymbol{\Lambda}_{b,r}$
denote the $r$th diagonal block of $\boldsymbol{\Lambda}_{a}$ and
$\boldsymbol{\Lambda}_{b}$, respectively, and $\lambda_{r}$ denote
the $r$th diagonal element of $\boldsymbol{\Lambda}_c$. We have
$\lambda_{r}\boldsymbol{\Lambda}_{a,r}
\boldsymbol{\Lambda}_{b,r}^T=\boldsymbol{I}, \forall r$.
\end{theorem}
\begin{proof}
Please refer to \cite{Lathauwer08}.
\end{proof}

From Theorem \ref{theorem2}, we know that if the following
conditions hold
\begin{align}
M_{\text{BS}}T'\geq \sum_{u=1}^U L_u^2\\
k'_{\boldsymbol{A}_Q}+k'_{\boldsymbol{A}_P}+k_{\boldsymbol{{S}}}\geq
2U+2 \label{g-Kruskal-condition}
\end{align}
then the essential uniqueness of the CP decomposition of
$\boldsymbol{\mathcal{Y}}$ in (\ref{CP}) can be guaranteed.
Following an analysis similar to our previous subsection, we can
arrive at the $k'$-ranks of $\boldsymbol{A}_Q$ and
$\boldsymbol{A}_P$ are equivalent to $U$ with probability one.
Therefore we only need $k_{\boldsymbol{{S}}}\geq 2$ in order
to satisfy the above generalized Kruskal's condition
(\ref{g-Kruskal-condition}). This condition can be easily
satisfied by assigning pairwise independent pilot symbol vectors
to users (provided $T\geq 2$).

%(note that $\boldsymbol{S}=
%[\boldsymbol{s}_1\phantom{0}\ldots\phantom{0}\boldsymbol{s}_U]$
%consists of $U$ column vectors, with each column vector
%corresponding to the pilot symbols of each user).

Since the proposed algorithm yields a canonical form of CP
decomposition represented as a sum of rank-one tensor components,
we need further explore the relationship between the true factor
matrices and the estimated factor matrices. We write
\begin{align}
\boldsymbol{\mathcal{X}}=&\sum_{r=1}^R(\boldsymbol{A}_r\boldsymbol{B}_r^T)\circ\boldsymbol{c}_r
\nonumber\\
=&\sum_{r=1}^R\sum_{j=1}^{L_R}\boldsymbol{A}_{r}[:,j]\circ\boldsymbol{B}_{r}[:,j]\circ\boldsymbol{c}_r
\nonumber\\
=&\sum_{l=1}^L \boldsymbol{a}_l
\circ\boldsymbol{b}_l\circ\boldsymbol{f}_l
\end{align}
where $L\triangleq\sum_{r=1}^R L_r$, $\boldsymbol{X}[:,j]$ denotes
the $j$th column of $\boldsymbol{X}$, $\boldsymbol{a}_l$ and
$\boldsymbol{b}_l$ denote the $l$th column of $\boldsymbol{A}$ and
$\boldsymbol{B}$, respectively, and
\begin{align}
\boldsymbol{f}_l=\boldsymbol{c}_r \quad \forall l\in
\left[\sum_{i=1}^{r-1}L_i+1, \sum_{i=1}^r L_i\right]
\end{align}
Define $\boldsymbol{F}\triangleq
[\boldsymbol{f}_1\phantom{0}\ldots\phantom{0}\boldsymbol{f}_L]$.
Clearly, $\boldsymbol{A}$, $\boldsymbol{B}$, and $\boldsymbol{F}$
are true factor matrices of $\boldsymbol{\mathcal{X}}$. The CP
decomposition of $\boldsymbol{\mathcal{X}}$ can also be expressed
as
\begin{align}
\boldsymbol{\mathcal{X}}=&\sum_{r=1}^R(\boldsymbol{\bar{A}}_r\boldsymbol{\bar{B}}_r^T)\circ\boldsymbol{\bar{c}}_r
\nonumber\\
=&\sum_{r=1}^R\sum_{j=1}^{L_r}(\boldsymbol{\bar{A}}_{r}[:,j]
\boldsymbol{\bar{B}}_{r}[:,j]^T)\circ(\lambda_{r}\boldsymbol{c}_r)
\nonumber\\
=&\sum_{r=1}^R\sum_{j=1}^{L_r}(\beta_{r,j}\boldsymbol{\bar{A}}_{r}[:,j]
\boldsymbol{\bar{B}}_{r}[:,j]^T)\circ(\beta_{r,j}^{-1}\lambda_{r}\boldsymbol{c}_r)
\nonumber\\
=&\sum_{l=1}^L
\boldsymbol{\tilde{a}}_{l}\circ\boldsymbol{\tilde{b}}_{l}\circ\boldsymbol{\tilde{f}}_{l}
\end{align}
where
\begin{align}
\boldsymbol{\tilde{a}}_{l}\triangleq
&\beta_{r,j}\boldsymbol{\bar{A}}_{r}[:,j] \quad
l=\sum_{i=1}^{r-1}L_r+j \nonumber \\
\boldsymbol{\tilde{b}}_{l}\triangleq
&\boldsymbol{\bar{B}}_{r}[:,j] \quad l=\sum_{i=1}^{r-1}L_r+j \nonumber\\
\boldsymbol{\tilde{f}}_{l} \triangleq &
\beta_{r,j}^{-1}\lambda_{r}\boldsymbol{c}_r \quad
l=\sum_{i=1}^{r-1}L_r+j \nonumber
\end{align}
Define
\begin{align}
\boldsymbol{\tilde{A}}\triangleq&
[\boldsymbol{\tilde{a}}_{1}\phantom{0}\boldsymbol{\tilde{a}}_{2}\phantom{0}\ldots\phantom{0}\boldsymbol{\tilde{a}}_{L}]
\nonumber \\
\boldsymbol{\tilde{B}}\triangleq&
[\boldsymbol{\tilde{b}}_{1}\phantom{0}\boldsymbol{\tilde{b}}_{2}\phantom{0}\ldots\phantom{0}\boldsymbol{\tilde{b}}_{L}]
\nonumber \\
\boldsymbol{\tilde{F}}\triangleq&
[\boldsymbol{\tilde{f}}_{1}\phantom{0}\boldsymbol{\tilde{f}}_{2}\phantom{0}\ldots\phantom{0}\boldsymbol{\tilde{f}}_{L}]
\nonumber
\end{align}
Clearly,
$(\boldsymbol{\tilde{A}},\boldsymbol{\tilde{B}},\boldsymbol{\tilde{F}})$
is an alternative solution which decomposes
$\boldsymbol{\mathcal{X}}$ into $L$ rank-one tensor components. It
is easy to verify that the true factor matrices
$(\boldsymbol{A},\boldsymbol{B},\boldsymbol{F})$ and the estimated
factor matrices
$(\boldsymbol{\tilde{A}},\boldsymbol{\tilde{B}},\boldsymbol{\tilde{F}})$
are related as follows:
\begin{align}
\boldsymbol{\tilde{A}}=&\boldsymbol{A}\boldsymbol{\Lambda}_1\boldsymbol{\Pi}
\label{factor-matrix-relation1}
\\
\boldsymbol{\tilde{B}}=&\boldsymbol{B}\boldsymbol{\Lambda}_2\boldsymbol{\Pi}
\label{factor-matrix-relation2}
\\
\boldsymbol{\tilde{F}}=&\boldsymbol{F}\boldsymbol{\Lambda}_3\boldsymbol{\Pi}
\label{factor-matrix-relation3}
\end{align}
where $\boldsymbol{\Pi}$ is a permutation matrix, and
\begin{align}
\boldsymbol{\Lambda}_1=&\boldsymbol{\Lambda}_a\boldsymbol{D}_{\beta}
\\
\boldsymbol{\Lambda}_2=&\boldsymbol{\Lambda}_b \\
\boldsymbol{\Lambda}_3=&\boldsymbol{D}_{\beta}^{-1}\boldsymbol{D}_{\lambda}
\end{align}
in which $\boldsymbol{D}_{\beta}$ is a diagonal matrix with its
$l$th ($l=\sum_{i=1}^{r-1}L_r+j$) diagonal element equal to
$\beta_{r,j}$, and
\begin{align}
\boldsymbol{D}_{\lambda}\triangleq\text{diag}(\lambda_1\boldsymbol{I}_{L_1},\ldots,
\lambda_R\boldsymbol{I}_{L_R})
\end{align}
where $\boldsymbol{I}_{n}$ is an $n\times n$ identity matrix. It
is easy to verify that
\begin{align}
\boldsymbol{\Lambda}_1\boldsymbol{\Lambda}_3\boldsymbol{\Lambda}_2^T
=\boldsymbol{\Lambda}_a\boldsymbol{D}_{\lambda}\boldsymbol{\Lambda}_b^T=\boldsymbol{I}
\label{factor-matrix-relation4}
\end{align}
since we have $\lambda_{r}\boldsymbol{\Lambda}_{a,r}
\boldsymbol{\Lambda}_{b,r}^T=\boldsymbol{I}, \forall r$.

%Consider the channel model and the coherent pilot transmission
%scheme described in Section \ref{syssection}.

%Let $\boldsymbol{\bar{\theta}}\triangleq
%[\bar{\theta}_1\phantom{0}\ldots\phantom{0}\bar{\theta}_{N_1}]$
%and $\boldsymbol{\bar{\phi}}\triangleq
%[\bar{\phi}_1\phantom{0}\ldots\phantom{0}\bar{\phi}_{N_2}]$, and

%In other words, the compressed sensing method can be directly
%applied to solve the multiuser channel estimation problem,

\section{A Direct Compressed Sensing-Based Channel Estimation Method}\label{sec:cs-method}
The multiuser channel estimation problem considered in this paper
can also be formulated as a sparse signal recovery problem by
exploiting the poor scattering nature of the mmWave channel,
without resorting to the CP decomposition. Such a direct
compressed sensing-based method is discussed in the following. Let
$\boldsymbol{Y}_{(3)}$ denote the mode-3 unfolding of the tensor
$\boldsymbol{\mathcal{Y}}$ defined in (\ref{CP}). We have
\begin{align}
\boldsymbol{Y}_{(3)}=&{\boldsymbol{S}}_L(\boldsymbol{A}_P\odot\boldsymbol{A}_Q)^T+\boldsymbol{W}_{(3)}
\nonumber\\
=&{\boldsymbol{S}}_L[\boldsymbol{\tilde{a}}_{\text{MS}}(\phi_1)\otimes
\boldsymbol{\tilde{a}}_{\text{BS}}(\theta_1)\phantom{0}\ldots\phantom{0}
\boldsymbol{\tilde{a}}_{\text{MS}}(\phi_L)\otimes
\boldsymbol{\tilde{a}}_{\text{BS}}(\theta_L)]^T +\boldsymbol{W}_{(3)}\nonumber\\
\stackrel{(a)}{=}&{\boldsymbol{S}}_L\boldsymbol{D}
\boldsymbol{\Sigma}^T(\boldsymbol{P}^T\otimes\boldsymbol{Q}^T)^T+\boldsymbol{W}_{(3)}
\end{align}
where $(a)$ comes from the mixed-product property:
$(\boldsymbol{A}\otimes\boldsymbol{B})(\boldsymbol{C}\otimes\boldsymbol{D})=(\boldsymbol{A}\boldsymbol{C})\otimes
(\boldsymbol{B}\boldsymbol{D})$, and
\begin{align}
\boldsymbol{\Sigma}\triangleq
[\boldsymbol{a}_{\text{MS}}(\phi_1)\otimes
\boldsymbol{a}_{\text{BS}}(\theta_1)\phantom{0}\ldots\phantom{0}
\boldsymbol{a}_{\text{MS}}(\phi_L)\otimes
\boldsymbol{a}_{\text{BS}}(\theta_L)] \nonumber
\end{align}
\begin{align}
\boldsymbol{D}\triangleq\text{diag}(\alpha_1,\ldots,\alpha_L)
\end{align}
Taking the transpose of $\boldsymbol{Y}_{(3)}$, we arrive at
\begin{align}
\boldsymbol{Y}_{(3)}^T=&(\boldsymbol{P}^T\otimes\boldsymbol{Q}^T)\boldsymbol{\Sigma}
\boldsymbol{D}\boldsymbol{O}^T{\boldsymbol{S}}^T+\boldsymbol{W}_{(3)}
\label{eqn-1}
\end{align}
The dictionary $\boldsymbol{\Sigma}$ is characterized by a number
of unknown parameters $\{\theta_l,\phi_l\}$ which need to be
estimated. To formulate the channel estimation as a sparse signal
recovery problem, we discretize the continuous parameter space
into an $N_1\times N_2$ two dimensional grid with each grid point
given by $\{\bar{\theta}_{i},\bar{\phi}_j\}$ for $i=1,\ldots, N_1$
and $j=1,\ldots,N_2$. Assume that the true parameters
$\{\theta_l,\phi_l\}$ lie on the two-dimensional grid. Hence
(\ref{eqn-1}) can be re-expressed as
\begin{align}
\boldsymbol{Y}_{(3)}^T=(\boldsymbol{P}^T\otimes\boldsymbol{Q}^T)\boldsymbol{\bar{\Sigma}}
\boldsymbol{\bar{D}}\boldsymbol{S}^T+\boldsymbol{W}_{(3)}
\end{align}
where $\boldsymbol{\bar{\Sigma}}$ is an overcomplete dictionary
consisting of $N_1\times N_2$ columns, with its $((i-1)N_1+j)$th
column given by $\boldsymbol{a}_{\text{MS}}(\bar{\phi}_i)\otimes
\boldsymbol{a}_{\text{BS}}(\bar{\theta}_j)$,
$\boldsymbol{\bar{D}}\in\mathbb{C}^{N_1N_2\times U}$ is a sparse
matrix obtained by augmenting ${\boldsymbol{D}}\boldsymbol{O}^T$
with zero rows. Let
$\boldsymbol{y}\triangleq\text{vec}(\boldsymbol{Y}_{(3)}^T)$ and
define
$\boldsymbol{\Phi}\triangleq(\boldsymbol{P}^T\otimes\boldsymbol{Q}^T)\boldsymbol{\bar{\Sigma}}$.
We have
\begin{align}
\boldsymbol{y}=({\boldsymbol S}\otimes\boldsymbol{\Phi})\boldsymbol{d}+\boldsymbol{w}
\label{eqn-2}
\end{align}
where $\boldsymbol{d}\triangleq\text{vec}(\boldsymbol{\bar{D}})$
is an unknown sparse vector, and
$\boldsymbol{w}\triangleq\text{vec}(\boldsymbol{W}_{(3)})$ denotes
the additive noise. We see that the channel estimation problem has
now been formulated as a conventional sparse signal recovery
problem. The problem can be further recast as an
$\ell_1$-regularized optimization problem
\begin{align}
\min_{\boldsymbol{d}}\quad
\|\boldsymbol{y}-({\boldsymbol S}\otimes\boldsymbol{\Phi})\boldsymbol{d}\|_2^2+\lambda\|\boldsymbol{d}\|_1
\label{L1-minimization}
\end{align}
and many efficient algorithms such as fast iterative
shrinkage-thresholding algorithm (FISTA) \cite{BeckTeboulle09} can
be employed to solve the above $\ell_1$-regularized optimization
problem. In practice, the true parameters may not be aligned on
the presumed grid. This error, also referred to as the grid
mismatch, leads to deteriorated performance. Finer grids can
certainly be used to reduce grid mismatch and improve the
reconstruction accuracy. Nevertheless, recovery algorithms may
become numerically instable and computationally prohibitive when
very fine discretized grids are employed.

\section{Computational Complexity Analysis}\label{sec:complexity-analysis}
We discuss the computational complexity of the proposed CP
decomposition-based method and its comparison with the direct
compressed sensing-based method. The computational task of our
proposed method involves solving the least squares problems
(\ref{updateA})--(\ref{updateC}) at each iteration and solving the
compressed sensing problem (\ref{Hu-estimation3}) after the factor
matrices are estimated. Let $\boldsymbol{A}=\boldsymbol{A}_Q$,
$\boldsymbol{B}=\boldsymbol{A}_P$,
$\boldsymbol{C}=\boldsymbol{S}_L$ in
(\ref{updateA})--(\ref{updateC}). Considering the update of
$\boldsymbol{A}_Q$, we have
$\boldsymbol{A}_Q^T=(\boldsymbol{V}^H\boldsymbol{V}+\mu\boldsymbol{I})^{-1}\boldsymbol{V}^H\boldsymbol{Y}_{(1)}^T$,
where
$\boldsymbol{V}\triangleq(\boldsymbol{S}^{(t)}\odot\boldsymbol{A}_{P}^{(t)})\in\mathbb{C}^{TT'\times
K}$ is a tall matrix as we usually have $TT'>K$. Noting that
$\boldsymbol{Y}_{(1)}^T\in\mathbb{C}^{TT'\times M_{\text{BS}}}$,
it can be easily verified that the number of flops required to
calculate $\boldsymbol{A}_Q^T$ is of order
$\mathcal{O}({K}T'{T}M_{\text{BS}}+K^2 T'T +K^3)$. $K$ is usually
of the same order of magnitude as the value of $L$. When $L$ is
small, the order of the dominant term will be
$\mathcal{O}(T'TM_{\text{BS}})$ which scales linearly with the
size of observed tensor $\boldsymbol{\mathcal{Y}}$. We can also
easily show that solving the least squares problems
(\ref{updateB}) and (\ref{updateC}) requires flops of order
$\mathcal{O}(T'TM_{\text{BS}})$ as well. To solve
(\ref{Hu-estimation3}), a fast iterative shrinkage-thresholding
algorithm (FISTA) \cite{BeckTeboulle09} can be used. The main
computational task associated with the FISTA algorithm at each
iteration is to evaluate a so-called proximal operator whose
computational complexity is of the order $\mathcal{O}(n^2)$, where
$n$ denotes the number of columns of the overcomplete dictionary.
For our case, the computational complexity is of order
$\mathcal{O}(N_1^{2}N_2^{2})$. Thus the overall computational
complexity is $\mathcal{O}(N_1^{2}N_2^{2}+T'TM_{\text{BS}})$.

For the direct compressed sensing-based method discussed in
Section \ref{sec:cs-method}, the main computational task
associated with the FISTA algorithm at each iteration is to
evaluate the proximal operator whose computational complexity, as
indicated earlier, is of the order $\mathcal{O}(n^2)$, where $n$
denotes the number of columns of the overcomplete dictionary. For
the compressed sensing problem considered in
(\ref{L1-minimization}), we have $n=N_{1}N_{2}U$. Thus the
required number of flops at each iteration of the FISTA is of
order $\mathcal{O}(N_1^{2}N_2^{2}U^2)$, which scales quadratically
with $N_{1}N_{2}U$. Note that the overcomplete dictionary
${\boldsymbol S}\otimes\boldsymbol{\Phi}$ in
(\ref{L1-minimization}) is of dimension $TT'M_{\text{BS}}\times
N_{1}N_{2}U$. In order to achieve a substantial overhead
reduction, the parameters $\{M_{\text{BS}}, T, T'\}$ are usually
chosen such that the number of measurements is far less than the
dimension of the sparse signal, i.e. $TT'M_{\text{BS}}\ll U
N_{1}N_{2}$. Therefore the compressed sensing-based method has a
higher computational complexity than the proposed CP
decomposition-based method.

%To acquire a high-resolution AoA-AoD estimation accuracy, a fine
%grid should be used and usually we have $N_1\geq N_{\text{BS}}$
%and $N_2\geq N_{\text{MS}}$.

%with each point specified by the angle of arrival and the angle of
%departure.

%on a two-dimensional AoA-AoD space

%In fact, performance of CPF is insensitive to the choices of the
%parameter when the tensor size is
%enlarged\cite{BazerqueMateos2013}.

%A suitable choice of $\mu$ depends on the noise level.

\section{Simulation Results}\label{sec:experiments}
We now present simulation results to illustrate the performance of
our proposed CP factorization-based method (referred to as CPF),
and its comparison with the direct compressed sensing-based method
(referred to as CS) discussed in Section \ref{sec:cs-method}. For
the CPF method, $\mu$ in (\ref{opt-1}) is chosen to be $3\times
10^{-3}$ throughout our experiments. In fact, empirical results
suggest that stable recovery performance can be achieved when
$\mu$ is set in the range $[10^{-3}, 10^{-2}]$. We consider a
system model consisting of a BS and $U$ MSs, with the BS employing
a uniform linear array of $N_{\text{BS}}=64$ antennas and each MS
employing a uniform linear array of $N_{\text{MS}}=32$ antennas.
We set $U=8$. The mmWave channel is assumed to follow a geometric
channel model with the AoAs and AoDs distributed in $[0,2\pi]$.
The complex gain $\alpha_{u,l}$ is assumed to be a random variable
following a circularly-symmetric Gaussian distribution
$\alpha_{u,l}\sim\mathcal{CN}(0,N_{\text{BS}}N_{\text{MS}}/\rho)$,
where $\rho$ is given by $\rho=(4\pi d f_c/c)^2$, here $c$
represents the speed of light, $d$ denotes the distance between
the MS and the BS, and $f_c$ is the carrier frequency. We assume
$d=50$ and $f_c=28$GHz. In our simulations, we investigate the
performance of the proposed method under two randomly generated
mmWave channels. For the first mmWave channel, the AoAs and AoDs
associated with the $U$ users are closely-spaced (see Fig.
\ref{fig2} (a)), while the AoAs and AoDs associated with the $U$
users are sufficiently separated for the other mmWave channel (see
Fig. \ref{fig2} (b)). The total number of paths is set to $L=13$
and the number of scatterers between each MS and the BS, $L_u$, is
set equal to one or two. The beamforming matrix $\boldsymbol{P}$
and the combining matrix $\boldsymbol{Q}$ are generated according
to the way described in Section \ref{sec:uniqueness-analysis}. The
pilot symbol matrix $\boldsymbol{S}$ is chosen from the codebook
of Grassmannian beamforming \cite{LoveHeath03} for $T=2$, while
for $T=3$, $T=4$ and $T=6$, $\boldsymbol{S}$ can be calculated by
the algorithm proposed in \cite{MedraDavidson14}. When $T=8$,
$\boldsymbol{S}$ is simply chosen as a DFT matrix.

\begin{figure*}[!t]
 \centering
\subfigure[The set of AoAs and AoDs associated with the first
channel]{\includegraphics[width=3.5in]{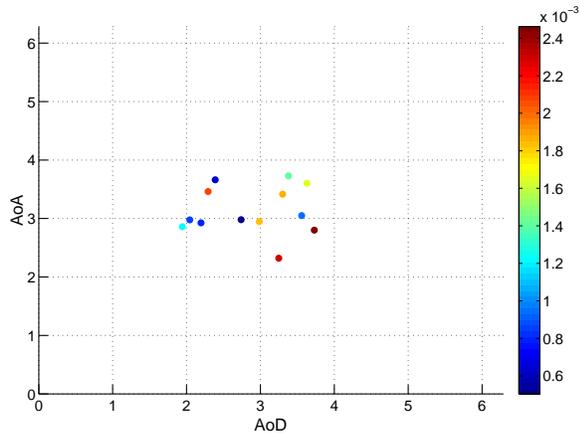} \label{fig2a}}
 \hfil
\subfigure[The set of AoAs and AoDs associated with the second
channel]{\includegraphics[width=3.5in]{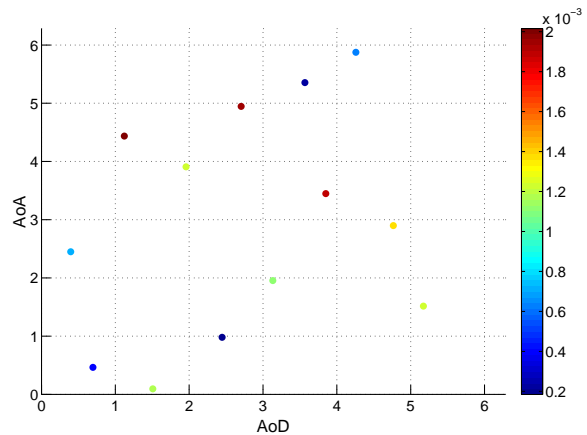}}
  \caption{Two sets of AoAs/AoDs realizations.}
   \label{fig2}
\end{figure*}

The estimation performance is evaluated by the normalized mean
squared error (NMSE) which is calculated as
\begin{align}
\text{NMSE}=\frac{\sum_{u=1}^U\|\boldsymbol{H}_u-\boldsymbol{\hat{H}}_u\|_F^2}{\sum_{u=1}^U\|\boldsymbol{H}_u\|_F^2}
\end{align}
where $\boldsymbol{\hat{H}}_u$ denotes the estimated channel. The
signal-to-noise ratio (SNR) is defined as the ratio of the signal
component to the noise component, i.e.
\begin{align}
\text{SNR}\triangleq
\frac{\|\boldsymbol{\mathcal{Y}}-\boldsymbol{\mathcal{W}}\|_F^2}{\|\boldsymbol{\mathcal{W}}\|_F^2}
\end{align}
where $\boldsymbol{\mathcal{Y}}$ and $\boldsymbol{\mathcal{W}}$
represent the received signal and the additive noise in
(\ref{CP}), respectively.

%For CP algorithm, the $\mu$ in (\ref{opt-1}) is chosen to be a
%fixed value from experience. It can always arrive at a good
%estimation as the discussion in \cite{BazerqueMateos2013}.

We first examine the channel estimation performance under
different SNRs. Fig. \ref{fig3} plots the estimation accuracy as a
function of SNR. Note that the compressed sensing method requires
to discretize the parameter space into a finite set of grid
points, and the true parameters may not lie on the discretized
grid. To illustrate the tradeoff between the estimation accuracy
and the computational complexity for the CS method, we employ two
different grids to discretize the continuous parameter space: the
first grid (referred to as Grid-I) discretizes the AoA-AoD space
into $64\times 32$ grid points, and the second grid (referred to
as Grid-II) discretizes the AoA-AoD space into $128\times 64$ grid
points. For our proposed CPF method, after the factor matrices are
estimated, a compressed sensing method is also used to estimate
each user's channel. Nevertheless, since the problem has been
decoupled into a set of single user's channel estimation problems
via CP factorization, the size of the overcomplete dictionary
involved in compressed sensing is now much smaller. Hence a finer
grid can be employed. In our simulations, we use a grid of
$256\times 128$ for our proposed method. Table \ref{table1} shows
that even using such a fine grid, our proposed method still
consumes much less average run times as compared with the CS
method which uses a grid of $128\times 64$. From Fig. \ref{fig3},
we see that our proposed method presents a clear performance
advantage over the CS method that employs the finer grid of the
two choices. The performance gain is possibly due to the following
two reasons. Firstly, our proposed method exploits intrinsic
multi-dimensional structure of the multiway data. Secondly, our
method benefits from the fact that the CP decomposition, which
serves as a critical step of our method, is essentially an
off-grid approach which does not suffer from grid mismatches. We
also observe that the CS method achieves a performance improvement
by employing a finer grid. Nevertheless, the required average
runtime increases drastically when a finer grid is used (see Table
\ref{table1}).

\begin{figure*}[!t]
 \centering
\subfigure[Channel
I]{\includegraphics[width=3.5in]{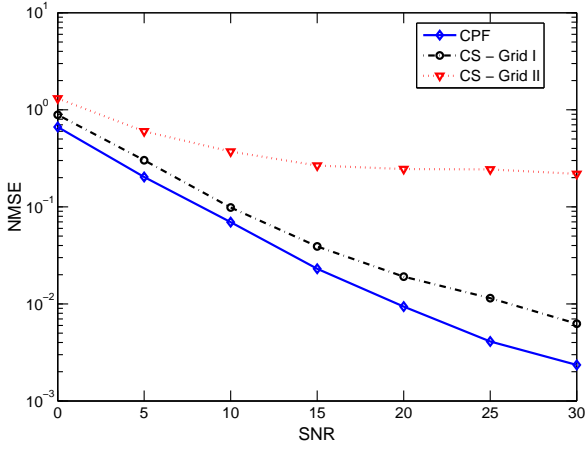}}
 \hfil
\subfigure[Channel
II]{\includegraphics[width=3.5in]{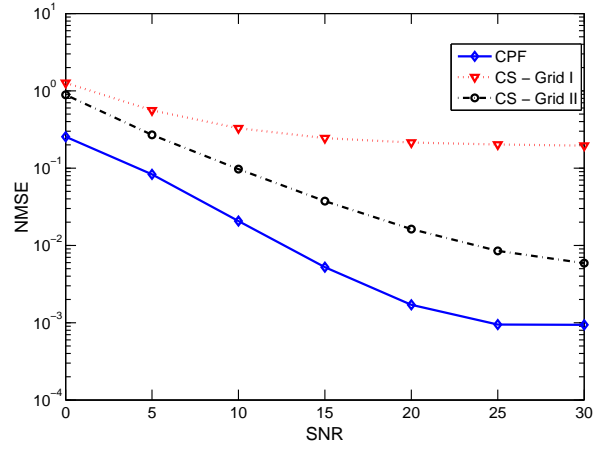}}
  \caption{NMSE versus SNR, $M_{BS}=16$, $T'=16$, and $T=4$.}
   \label{fig3}
\end{figure*}

Next, we examine how the estimation performance depends on the
parameters $T$, $M_{\text{BS}}$ and $T'$. Fig. \ref{fig4} shows the
NMSEs of respective algorithms as $T$ varies from 2 to 8, and the
other two parameters $T'$ and $M_{\text{BS}}$ are fixed to be
$T'=16$ and $M_{\text{BS}}=16$. Since $T'>U$ and $M_{BS}>U$, the
generalized Kruskal's condition (\ref{g-Kruskal-condition}) can be
satisfied when $k_{\boldsymbol{S}}\geq2$, that is, $T\geq 2$. From
Fig. \ref{fig4}, we see that when $T>2$, our proposed method is able to
provide a reliable channel estimate. This result roughly coincides
with our previous analysis regarding the uniqueness of the CP
decomposition. We also observe that better estimation performance
can be achieved for the latter mmWave channel. This is expected
since the mutual coherence of the factor matrices
$\boldsymbol{A}_Q$, $\boldsymbol{A}_P$ becomes lower as the
AoAs/AoDs are more sufficiently separated. As a result, the CP
factorization can be accomplished with a higher accuracy.

%Our proposed method again outperforms the compressed sensing
%method.

\begin{figure*}[!t]
\centering \subfigure[Channel
I]{\includegraphics[width=3.5in]{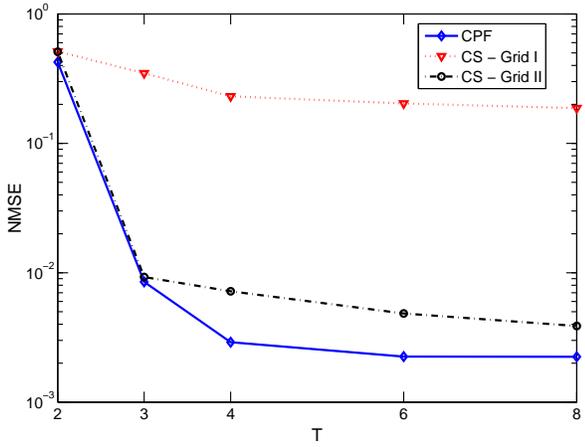}} \hfil
\subfigure[Channel
II]{\includegraphics[width=3.5in]{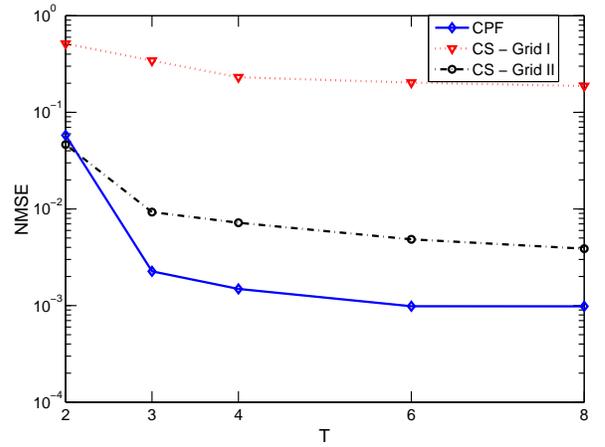}}
\caption{NMSE versus $T$, $M_{BS}=16$, $T'=16$, SNR=30dB.}
\label{fig4}
\end{figure*}

Fig. \ref{fig5} depicts the NMSEs of respective algorithms as a
function of $M_{\text{BS}}$, where we set $T'=16$, $T=4$, and
$\text{SNR}=30\text{dB}$. To satisfy (\ref{g-Kruskal-condition}),
it is easy to know that $M_{\text{BS}}$ should be greater than or
equal to 11. From Fig. \ref{fig5}, we see that our simulation
results again roughly corroborate our analysis: the proposed
method provides a decent estimation accuracy when the generalized
Kruskal's is satisfied, i.e. $M_{\text{BS}}> 11$. Also, our
proposed method outperforms the compressed sensing method by a
considerable margin. In Fig. \ref{fig6}, we plot the estimation
accuracy of respective algorithms as a function of $T'$, where we
set $M_{\text{BS}}=16$, $T=4$, and $\text{SNR}=30\text{dB}$.
Similar conclusions can be made from this figure.

\begin{figure*}[!t]
\centering \subfigure[Channel
I]{\includegraphics[width=3.5in]{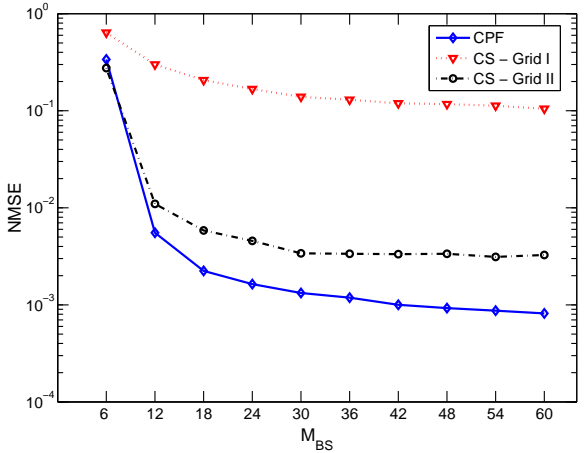}} \hfil
\subfigure[Channel
II]{\includegraphics[width=3.5in]{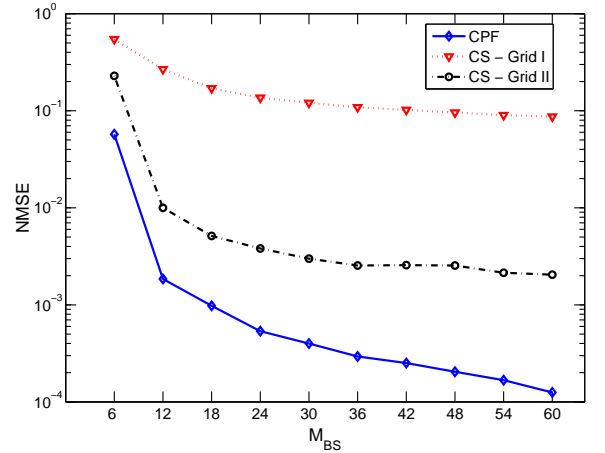}}
\caption{NMSE versus $M_{BS}$, $T^{\prime}=16$, $T=4$, SNR=30dB.}
\label{fig5}
\end{figure*}

\begin{figure*}[!t]
\subfigure[Channel
I]{\includegraphics[width=3.5in]{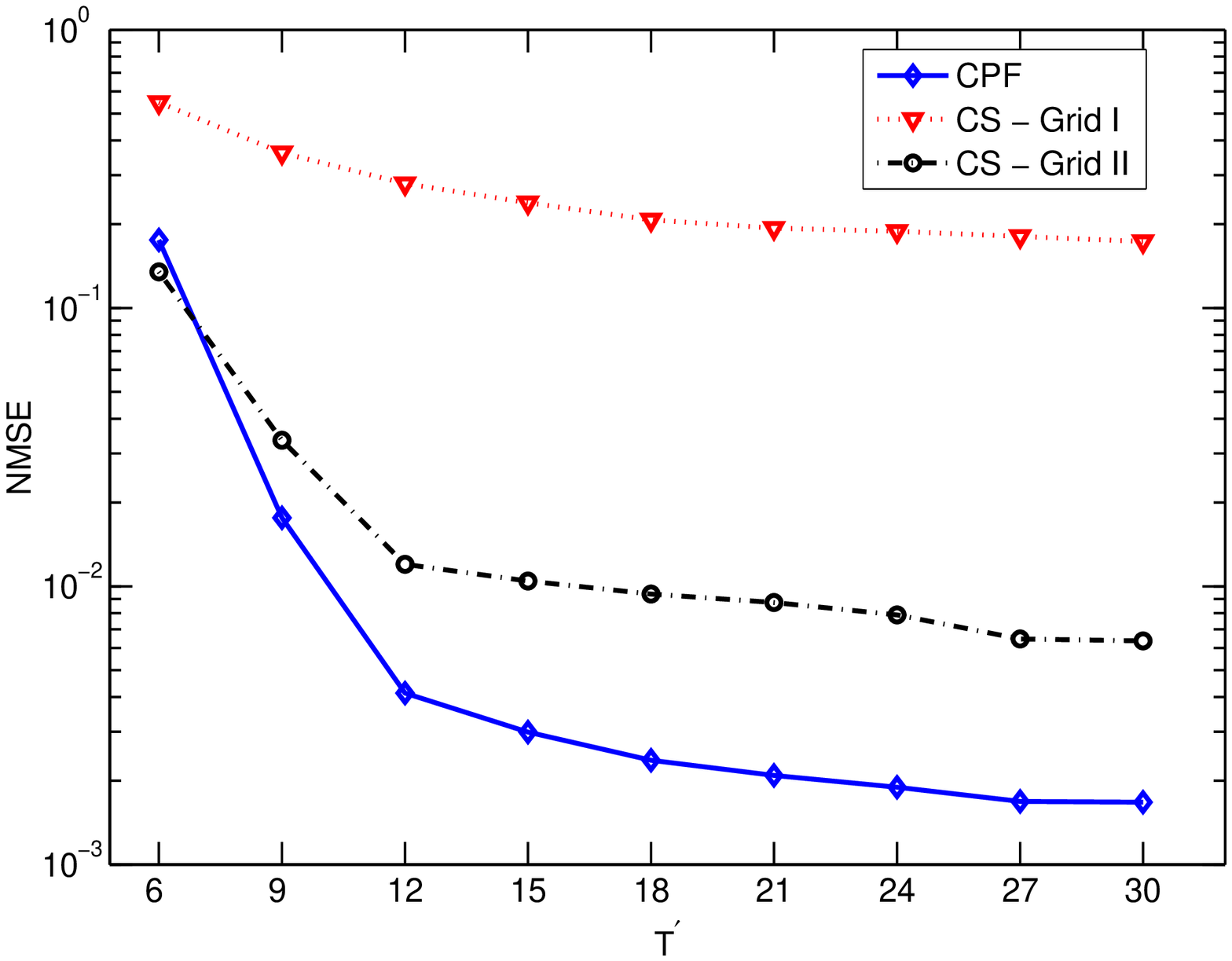}}
\hfil \subfigure[Channel
II]{\includegraphics[width=3.5in]{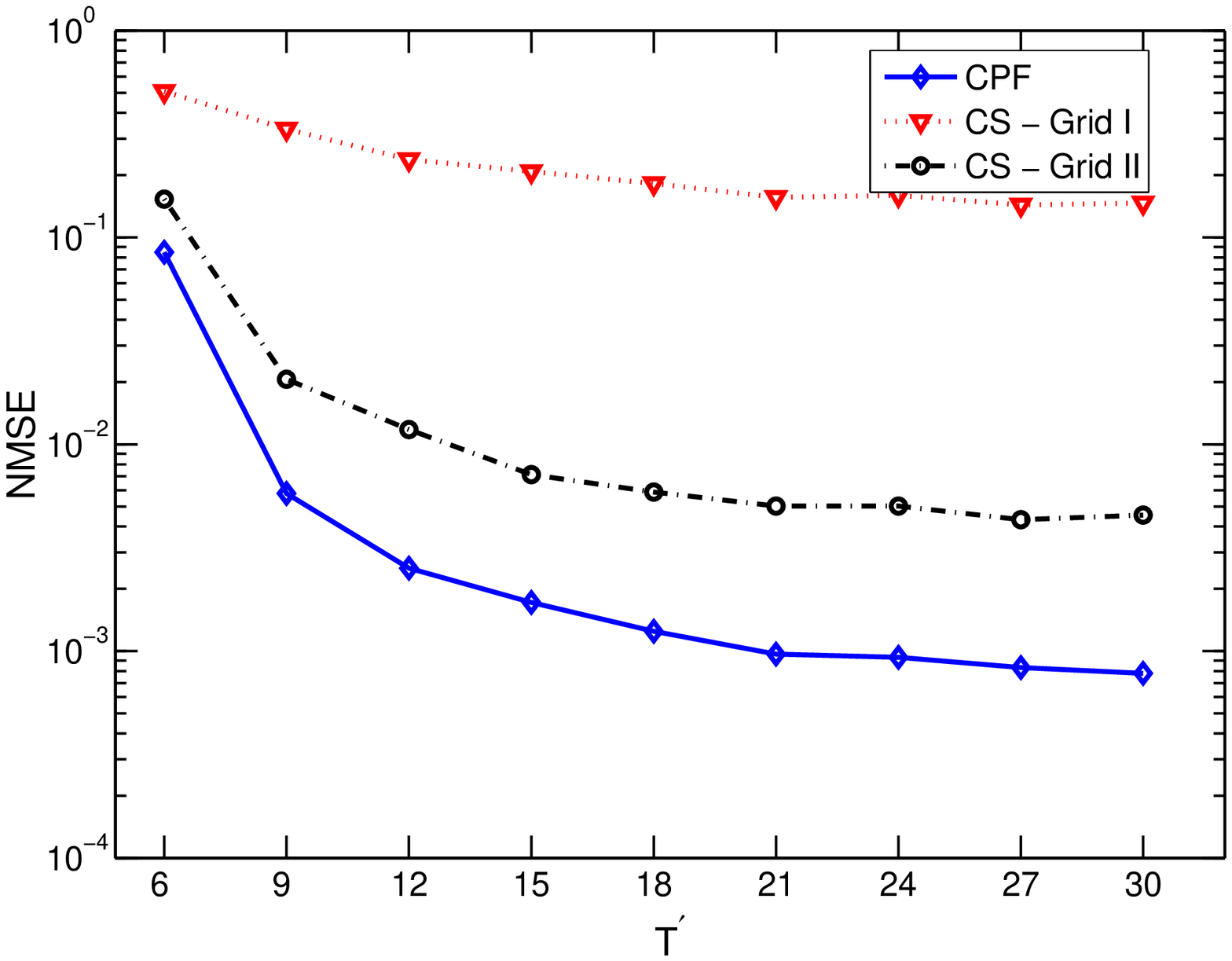}}
\caption{NMSE versus $T^{\prime}$, $M_{BS}=16$, $T=4$, SNR=30dB.}
\label{fig6}
\end{figure*}

Table \ref{table1} shows the average run times of our proposed
method and the compressed sensing method. We see that the
computational complexity of the compressed sensing method grows
dramatically as the dimension of the grid increases. Our proposed
method is more computationally efficient than the compressed
sensing method. It takes similar run times as the direct
compressed sensing method which employs the coarser grid of the
two choices, meanwhile achieving a better estimation accuracy than
the compressed sensing method that uses the finer grid.

\begin{table}[!t]
\caption{Average run times of respective algorithms, $T'=16$,
$M_{BS}=16$, $T=4$}
\renewcommand{\arraystretch}{1.3}
\centering
\begin{tabular}{c|c||c|c|c|c}
\hline
\multirow{2}*{\bfseries{ALG}}&\multirow{2}*{\bfseries{Grid}}
&\multicolumn{2}{| c |}{\bfseries{ NMSE}}&\multicolumn{2}{| c }{\bfseries{ Average Run Time(s)}}\\
\cline{3-6}
 & &Channel I&Channel II&Channel I&Channel II\\
\hline
{\multirow{2}*{CS}}&$64\times32$&$2.5e-1$&$2.3e-1$&$16.5$&$11$\\
\cline{2-6}
 &$128\times64$&$6.7e-3$&$6.4e-3$&$270$&$220$\\
\hline
CPF&-&$2.7e-3$&$1.5e-3$&$23$&$19$\\
\hline
\end{tabular}
\label{table1}
\end{table}

\section{Conclusions}\label{sec:conclusion}
We proposed a layered pilot transmission scheme and a
CANDECOMP/PARAFAC (CP) decomposition-based method for uplink
multiuser channel estimation in mm-Wave MIMO systems. The joint
uplink multiuser channel estimation was formulated as a tensor
decomposition problem. The uniqueness of the CP decomposition was
investigated for both the single-path geometric model and the
general geometric model. The conditions for the uniqueness of the
CP decomposition shed light on the design of the beamforming
matrix and the combining matrix, and meanwhile provide general
guidelines for choosing the system parameters. The proposed method
is able to achieve an additional training overhead reduction as
compared with a conventional scheme which separately estimates
multiple users' channels. Simulation results show that our
proposed method presents a clear performance advantage over the
compressed sensing method, and meanwhile achieving a substantial
computational complexity reduction.

\bibliography{newbib}
\bibliographystyle{IEEEtran}

\end{document}